\documentclass[a4paper,UKenglish,cleveref,thm-restate]{lipics-v2021}

\pdfoutput=1 
\hideLIPIcs  

\title{First-Order Logic and Twin-Width for Some Geometric Graphs}

\author{Colin Geniet}{Discrete Mathematics Group, Institute for Basic Science (IBS), Daejeon, South Korea}{colin@ibs.re.kr}{https://orcid.org/0000-0003-4034-7634}{Supported by the Institute for Basic Science (IBS-R029-C1).}

\author{Gunwoo Kim}{School of Computing, KAIST, Daejeon, South Korea \and Discrete Mathematics Group, Institute for Basic Science (IBS), Daejeon, South Korea}{gunwoo.kim@kaist.ac.kr}{}{Supported by the Institute for Basic Science (IBS-R029-C1).}

\author{Lucas Meijer}{Department of Information and Computing Sciences, Utrecht University, The Netherlands}{l.meijer2@uu.nl}{https://orcid.org/0009-0002-4901-5249}{Supported by the Netherlands Organisation for Scientific Research (NWO) under project no. VI.Vidi.213.150.}

\authorrunning{C.\ Geniet, G.\ Kim, and L.\ Meijer}
\Copyright{Colin Geniet, Gunwoo Kim, and Lucas Meijer}



\category{} 

\relatedversion{} 

\acknowledgements{%
    The authors are extremely grateful to Eunjung Kim for suggesting the questions treated in this work and for very insightful discussions on these topics,
    as well as Sebastian Wiederrecht for stimulating discussions during the early stages of this project.
    We would also like to thank Édouard Bonnet for very helpful explanations regarding twin-width of segment graphs.
}

\nolinenumbers 

\usepackage{mathtools}
\usepackage{todonotes}
\usepackage{comment}
\usepackage{tikz}
\usepackage{circuitikz}
\usepackage{apxproof}
\usepackage[dvipsnames]{xcolor}

\usepackage{xspace}
\usepackage{scalerel}

\newtheorem{fact}[theorem]{Fact}
\theoremstyle{remark}

\newcommand{\eqdef}{\coloneqq}
\newcommand{\Fc}{\mathcal{F}}
\newcommand{\Cc}{\mathcal{C}}
\newcommand{\Dc}{\mathcal{D}}
\newcommand{\Ec}{\mathcal{E}}
\newcommand{\Gc}{\mathcal{G}}
\newcommand{\Hc}{\mathcal{H}}
\newcommand{\Kc}{\mathcal{K}}
\newcommand{\Pc}{\mathcal{P}}
\newcommand{\Rc}{\mathcal{R}}
\newcommand{\Mc}{\mathcal{M}}
\newcommand{\Vc}{\mathcal{V}}
\newcommand{\tww}{\mathsf{tww}}
\newcommand{\mw}{\mathsf{mw}}
\newcommand{\N}{\mathbb{N}}
\newcommand{\Z}{\mathbb{Z}}
\newcommand{\R}{\mathbb{R}}
\newcommand{\lex}{\mathsf{lex}}

\newcommand{\fpt}{\ensuremath{\mathsf{FPT}}\xspace}
\newcommand{\aw}{\ensuremath{\mathsf{AW}[*]}\xspace}
\newcommand{\Perm}{\mathfrak{S}}

\newcommand{\urc}{{\scalerel*{\tikz \draw (0,1ex) -- (1ex,1ex) -- (1ex,0);}{1}}}
\newcommand{\ulc}{{\scalerel*{\tikz \draw (0,0) -- (0,1ex) -- (1ex,1ex);}{1}}}
\newcommand{\lrc}{{\scalerel*{\tikz \draw (0,0) -- (1ex,0) -- (1ex,1ex);}{1}}}
\newcommand{\llc}{{\scalerel*{\tikz \draw (0,1ex) -- (0,0) -- (1ex,0);}{1}}}

\newcommand{\seg}{\mathsf{seg}}

\newcommand{\Lend}{\mathsf{left}}
\newcommand{\Rend}{\mathsf{right}}
\newcommand{\Vseg}{\mathsf{Ver}}
\newcommand{\Hseg}{\mathsf{Hor}}
\newcommand{\Xco}{\pi_x}
\newcommand{\Yco}{\pi_y}

\newcommand{\adj}{Adj}
\newcommand{\inc}{Inc}
\newcommand{\floor}[1]{\lfloor #1 \rfloor}

\newcommand{\hormin}{\ensuremath{\mathsf{Hmin}}}
\newcommand{\vermin}{\ensuremath{\mathsf{Vmin}}}

\renewcommand{\emptyset}{\varnothing}

\begin{document}

\maketitle

\begin{abstract}
    For some geometric graph classes, tractability of testing first-order formulas is precisely characterised by the graph parameter \emph{twin-width}.
    This was first proved for interval graphs among others in [BCKKLT, IPEC '22], where the equivalence is called \emph{delineation}, and more generally holds for circle graphs, rooted directed path graphs, and $H$-graphs when $H$ is a forest.
    Delineation is based on the key idea that geometric graphs often admit natural vertex orderings, allowing to use the very rich theory of twin-width for ordered graphs.
    
    Answering two questions raised in their work, we prove that delineation holds for intersection graphs of non-degenerate axis-parallel unit segment graphs, but fails for visibility graphs of 1.5D terrains.
    We also prove delineation for intersection graphs of circular arcs.
\end{abstract}

\section{Introduction}
The \emph{first-order model checking} problem asks, given a graph~$G$ and a first-order (FO) formula~$\phi$, to test whether~$G$ satisfies~$\phi$.
In general, this is a difficult problem, known to be \aw-hard~\cite{downey1996queries},
and subsumes problems such as finding an independent set or dominating set of a specified size, or any given induced subgraph in~$G$.
For a graph class~$\Cc$, one may ask whether this FO model checking problem admits a \emph{fixed parameter tractable} (FPT) algorithm,
i.e.\ one running in time $f(|\phi|) \cdot |G|^{O(1)}$ for some (computable) function~$f$.

Classically, this problem has been studied in sparse graph classes, such as bounded degree graphs~\cite{seese1996linear} and more generally nowhere dense classes~\cite{grohe2017deciding}, which admit such FPT algorithms.
A more recent and incomparable development has come from the graph parameter \emph{twin-width} introduced by Bonnet, Kim, Thomassé, and Watrigant~\cite{Bonnet2022twinwidth1}:
if~$\Cc$ is a class of bounded twin-width, and witnesses of this (so-called \emph{contraction sequences}) can be efficiently computed,
then FO model checking in~$\Cc$ is FPT.

Since bounded twin-width and bounded degree are incomparable, a class with FPT FO model checking need not have bounded twin-width.
Yet a number of remarkable cases exist where twin-width is exactly the boundary of efficient FO model checking:
a fundamental result of Bonnet, Giocanti, Ossona de Mendez, Simon, Thomassé, and Toruńczyk~\cite{twin-width4} is that
for any hereditary (i.e.\ closed under induced subgraphs) class~$\Cc$ of \emph{ordered graphs} (i.e.\ graphs given with a total ordering of vertices, which formulas can query), model checking in~$\Cc$ is FPT if and only if~$\Cc$ has bounded twin-width\footnote{Under the assumption $\fpt \neq \aw$, or equivalently that model checking for \emph{all graphs} is not FPT.}.
By reducing to ordered graphs using a `nice' vertex ordering, similar results were proved for interval graphs and rooted directed path graphs by Bonnet, Chakraborty, Kim, Köhler, Lopes, and Thomassé~\cite{bonnet2022delineation}, who called this kind of equivalence \emph{delineation}.
It also holds for circle graphs~\cite{hlineny2022geometric}, $H$-graphs for any fixed forest~$H$~\cite{bonomo2025noncrossinghgraphs}, and tournaments~\cite{geniet2023tournaments}.
We continue this line of delineation results for geometric intersection graphs.
\begin{theorem}\label{thm:main}
    In each of the following cases, assuming that~$\Cc$ is hereditary and $\fpt \neq \aw$, the class~$\Cc$ has bounded twin-width if and only if FO model checking in~$\Cc$ is FPT:
    \begin{enumerate}
        \item $\Cc$ is a subclass of intersection graphs of arcs in the circle,
        \item $\Cc$ is a subclass of intersection graphs of non-degenerate, axis-parallel unit segments
        (where for algorithmic purposes, graphs in~$\Cc$ are given by their geometric representation), or
        \item more generally, $\Cc$ is a subclass of intersection graphs of non-degenerate, axis-parallel segments with lengths in a fixed finite set $L \subset \R_+$ (again given by geometric representations).
    \end{enumerate}
\end{theorem}
The question of delineation was asked in~\cite{bonnet2022delineation} for axis-parallel unit segment graphs, and in~\cite{bonomo2025noncrossinghgraphs} for circular arc graphs.

For axis parallel segments, `non-degenerate' in the statement means we do not allow two horizontal (resp.\ two vertical) segments to intersect.
We require geometric representations of segment graphs to be given, as computing them is NP-hard~\cite{mustata2013unitgrid}.
This is not an issue in circular arc graphs, where representations can be found in polynomial time~\cite{tucker1980circular}.

To complement \cref{thm:main}, we also give a number of examples of non-delineated classes of geometric graphs.
First, relaxing any of the assumptions from the last item of \cref{thm:main} breaks delineation:
\begin{restatable}{theorem}{nondelinsegments}\label{thm:seg-non-delineated}
    There is a hereditary graph class~$\Cc$ with FPT FO model checking, but unbounded twin-width,
    and which can be realised as intersection graphs of any of the following types of segments:
    \begin{enumerate}
        \item \label{item:degen} degenerate axis-parallel unit segments,
        \item \label{item:twolengths} axis-parallel segments in general position where segment lengths take only two values (but these values can vary depending on the graph~$G \in \Cc$),
        \item \label{item:epsilon-lengths} axis-parallel segments in general position with lengths in $[1,1+\epsilon]$ for any fixed~$\epsilon > 0$, or
        \item \label{item:epsilon-directions} unit segments in general position, even when segment slopes are restricted to `vertical', `horizontal plus $\epsilon$', and `horizontal minus $\epsilon$'.
    \end{enumerate}
\end{restatable}
It was already known from \cite[Figure~8]{bonnet2022delineation} that axis-parallel segments with only two lengths give non-delineated graph classes (case~\ref{item:twolengths}).
We still include it in the result as, amusingly, a single graph class~$\Cc$ is an example of non-delineation for all four cases of \cref{thm:seg-non-delineated}.
Our non-delineated class~$\Cc$ differs significantly from the one constructed in \cite[Figure~8]{bonnet2022delineation}, which was an encoding of all cubic graphs.
Graphs in~$\Cc$ consist of two half-graphs, joined by paths of slightly more than constant length along an arbitrary permutation.
We prove that they have unbounded twin-width, but bounded \emph{merge-width} (a generalisation of twin-width capturing bounded expansion classes proposed by Dreier and Toruńczyk~\cite{merge-width}).
The latter implies monadic dependence and FPT FO model checking.

A variant of this construction shows that delineation also fails for visibility graphs of 1.5D-terrains (i.e.\ polygonal curves formed by an $x$-monotone sequence of points), answering negatively a question of~\cite{bonnet2022delineation}.
\begin{restatable}{theorem}{terrainnotdelin}\label{thm:terrain-not-delin}
    There is a class~$\Cc$ of 1.5D-terrain visibility graphs, whose hereditary closure allows FPT FO model checking, but which has unbounded twin-width.
\end{restatable}

Finally, delineation of circular arc graphs (case~1 of \cref{thm:main}) can be stated as a result on \emph{$H$-graphs} (i.e.\ the intersection graphs of connected subspaces in the topological realisation of~$H$), taking~$H$ to be a cycle.
It is also known that $H$-graphs are delineated when~$H$ is any fixed forest~\cite{bonomo2025noncrossinghgraphs},
and it seems likely that the two proofs could be combined to obtain delineation when~$H$ is a fixed graph with at most one cycle.
We however observe that delineation fails for any more complex~$H$:
\begin{theorem}\label{thm:Hgraphs-non-delin}
    There is a hereditary class~$\Cc$ of graphs with FPT FO model checking, but unbounded twin-width and unbounded merge-width,
    such that if~$H$ is any connected graph with at least~2 cycles, then graphs in~$\Cc$ are $H$-graphs.
\end{theorem}

\paragraph*{Overview of the proof of \cref{thm:main}}
Let~$\Dc$ be the graph class we are interested in, e.g.\ the class of all circular-arc graphs.
Our goal is the following:
for any hereditary\footnote{The hereditary assumption is necessary to avoid classes in which model checking is easy for trivial reasons, e.g.\ graphs of the form: an arbitrary graph on~$n$ vertices, plus~$2^n$ isolated vertices as padding.} subclass $\Cc \subseteq \Dc$, there is an FPT model checking algorithm in~$\Cc$ if and only if~$\Cc$ has bounded twin-width.
Inspired by~\cite{bonnet2022delineation,hlineny2022geometric}, we obtain such an equivalence by proving a dichotomy of the following form, for any subclass~$\Cc \subseteq \Dc$:
\begin{itemize}
    \item Either we can compute in polynomial time a contraction sequence for any $G \in \Cc$, witnessing that~$G$ has bounded twin-width.
    Then~\cite{Bonnet2022twinwidth1} gives FPT model checking in~$\Cc$.
    \item Or~$\Cc$ is \emph{monadically independent}, which informally means that graphs in~$\Cc$ can be used to encode arbitrary graphs through FO logic (see \cref{sec:prelim:FO} for the definition).
    In that case, and if~$\Cc$ additionally is hereditary, then FO model checking on the class of all graphs reduces to FO model checking on~$\Cc$~\cite{flip-breakability}, proving that the latter is hard for the parametrised complexity class \aw~\cite{downey1996queries}.
\end{itemize}
Let us point out that if~$\Cc$ has bounded twin-width, then it is monadically dependent (the opposite of monadically independent).
Thus a subclass $\Cc \subseteq \Dc$ has bounded twin-width if and only if it is monadically dependent;
this is the precise definition of \emph{delineation} in~\cite{bonnet2022delineation}.
The above dichotomy is a slight strengthening of delineation, additionally requiring to be able to efficiently find contraction sequences.
When~$\Cc$ is hereditary, and assuming $\fpt \neq \aw$, this strengthening gives the equivalence of $\Cc$ having bounded twin-width, $\Cc$ being monadically dependent, and $\Cc$ admitting an FPT model checking algorithm.
Let us point out that it is a major conjecture in finite model theory that the last two points are always equivalent:
\begin{conjecture}[{\cite[Conjecture~8.2]{gajarsky2020FO}}]
    For any hereditary class~$\Cc$, there is an FPT FO model checking algorithm for graphs in~$\Cc$ if and only if~$\Cc$ is monadically dependent.
\end{conjecture}
Our results confirm special cases of this, where twin-width explains the equivalence.

Coming back to the proof overview, let us explain how the aforedescribed dichotomy is achieved.
The starting point is the characterisation of twin-width through obstructions in matrices~\cite{Bonnet2022twinwidth1}:
given any 0--1 matrix~$M$, one can efficiently find either a contraction sequence for~$M$ (witness of small twin-width), or a \emph{grid}.
Given the geometric representation of the graph~$G$, we construct some auxiliary matrix~$M_G$.
For instance, for circular arc graphs, we identify the circle with $\{1,\dots,n\}$ modulo~$n$, and the matrix~$M$ has a~1 at position~$(i,j)$ whenever there is an arc clockwise from~$i$ to~$j$.
We then prove two facts:
\begin{itemize}
    \item First, $G$ can be reconstructed from~$M_G$ by a sequence of operations that preserve twin-width.
    Thus if~$M_G$ has bounded twin-width, then so does~$G$. This is efficient: a contraction sequence for~$G$ can quickly be computed from one for~$M_G$.
    \item Second, if~$M_G$ contains a sufficiently large grid, then~$G$ contains a complex structure called \emph{transversal pair}.
    It is known that if graphs in~$\Cc$ contain arbitrarily large transversal pairs, then~$\Cc$ is monadically independent~\cite{bonnet2022delineation}.
    To prove this, we assume the geometric representation of~$G$ we use is \emph{minimised} in an appropriate sense.
\end{itemize}
Thus either one can find contraction sequences for the auxiliary matrices $\{M_G : G \in \Cc\}$, and thus also for~$\Cc$,
or these matrices have arbitrarily large grids, hence~$\Cc$ contains arbitrarily large transversal pairs and is monadically independent.
The case of APUS graphs requires one additional step: we construct \emph{local} auxiliary matrices, corresponding to 1-by-1 squares in the segment representation.
If any such matrix has a grid, then we find a transversal pair; if all local matrices have bounded twin-width, then so does the full graph.

\paragraph*{Open problems}
Continuing on the theme of intersection graphs of objects in the plane, we ask whether unit-disk graphs are delineated.
The difficulty in adapting our technique to this case may be to define and use the right notion of minimisation.
Disks with a fixed, finite set of allowed radii should behave similar to unit disks,
but we expect that allowing two radii with unbounded ratio, or arbitrary radii in $[1,1+\epsilon]$ causes delineation to fail with constructions similar to that of \cref{thm:seg-non-delineated}.
One may also ask about intersection graphs of squares of rectangles.
We expect our techniques to easily apply to axis-aligned unit squares, but wonder whether delineation still holds without the axis-aligned constraint.

One may also consider other parameters than twin-width.
Since our non-delineated constructions in \cref{thm:seg-non-delineated,thm:terrain-not-delin} use merge-width to obtain FPT model checking,
it is sensible to look for geometric graph classes in which merge-width characterises monadic dependence and FPT model checking.
Tackling this question would likely require characterising merge-width through some obstructions, a difficult open problem.
One can observe that this equivalence does \emph{not} hold for axis-parallel segments:
the construction of \cite[Figure~8]{bonnet2022delineation} shows that for any bipartite graph $H = (V,W,E)$,
one can obtain as axis-parallel segment graph the 2-subdivision of~$H$, with a full biclique added between~$V$ and~$W$.
Taking~$H$ to be graphs with girth and minimum degree going to infinity, this class is monadically dependent, but has unbounded merge-width.

\section{Preliminaries}\label{sec:prelim}
For $n \in \N$, we denote by~$[n]$ the interval of integers~$\{1,\dots,n\}$.

\subsection{Graphs and relational structures}
For a graph $G = (V,E)$, we write~$V(G)=V$ and~$E(G)=E$ for its vertex and edge sets.
The set of neighbours of~$v \in V(G)$ is ~$N(v)$. For $X \subseteq V$, $G[X]$ is the subgraph induced by~$X$.
When~$<$ is a linear ordering of~$V(G)$, we denote by $\adj(G,<)$ the adjacency matrix with rows and columns are ordered by~$<$.
Similarly, if~$<_V$ and~$<_E$ are orderings of~$V$ and~$E$,
then $\inc(G,<_V,<_E)$ is the incidence matrix with rows and columns ordered by~$<_V$ and~$<_E$.

More generally, a \emph{binary relational structure} (or simply binary structure) $(V,E_1,\dots,E_\ell)$
consists of a vertex set~$V$, and a number of relations $E_i \subseteq V^2$.
It can be viewed as a directed graph with~$\ell$ kinds of edges,
where there may be several edges of different kinds between two vertices (but not multiple edges of the same kind).
For binary structures, each relation has its own adjacency matrix, denoted by~$\adj(E_i,<)$.
An \emph{ordered binary structure} is a binary structure in which one of the relations is a linear ordering of the vertex set, normally denoted by~$<$. The structure is thus of the form $(V,E_1,\dots,E_\ell,<)$.
A special case is \emph{ordered graphs}~$(V,E,<)$, where~$(V,E)$ is a graph and~$<$ a linear ordering of~$V$.

\subsection{Grids}
\label{subsec:grids}
The following notion of grids comes from~\cite{GuillemotMarx}.
Given a point~$p \in \R^2$, we write $\Xco(p),\Yco(p)$ for its $x$- and $y$-coordinates.
A \emph{$t$-grid} in the plane is a family of points $(p_{i,j})_{i,j \in [t]}$ satisfying the following for any $i,i',j,j' \in [t]$:
\begin{align*}
    & \text{if $i<i'$, then $\Xco(p_{i,j}) < \Xco(p_{i',j'})$}
    \qquad \text{and} \qquad \text{if $j<j'$, then $\Yco(p_{i,j}) < \Yco(p_{i',j'})$.}
\end{align*}
Observe then that there are intervals $X_1 < \dots < X_t$ and $Y_1 < \dots < Y_t$ such that $p_{i,j} \in X_i \times Y_j$.
Indeed, take~$X_i$ to be the smallest interval containing $\Xco(p_{i,j})$ for all~$j$, and similarly with~$Y_i$.

This extends to 0--1 matrices, interpreting a~`1' at position~$(i,j)$ in the matrix~$M$ as a point with coordinates~$(i,j)$.
We say that~$M$ contains a $t$-grid if the point set defined by the~`1's in~$M$ contains a $t$-grid in the previous sense.
Then there are partitions into intervals $\Rc = \{X_1 < \dots < X_t\}$ and $\Cc = \{Y_1 < \dots < Y_t\}$ of the rows and columns as above.
The ordering of rows and columns is crucial in this definition, which is why we make it explicit when defining an adjacency matrix~$\adj(G,<)$.

The Ramsey property for grids easily follows from the Ramsey theorem on bicliques:
\begin{theorem}[Ramsey theorem for grids]\label{thm:grid-ramsey}
    There is a function~$R(t,k)$ such that if~$\Pc$ is a point set containing a $R(t,k)$-grid,
    and $\lambda : \Pc \to [k]$ is any colouring, then~$\Pc$ contains $t$-grid consisting only of points of colour~$c$ for some~$c \in [k]$.
\end{theorem}
\begin{proof}
    Choose~$R(t,k)$ to be the bipartite Ramsey bound.
    For $T = R(t,k)$, consider a grid $(p_{i,j})_{i,j \in [T]}$.
    Construct a biclique with vertices $\{u_i\}_{i \in [T]},\{v_i\}_{i \in [T]}$, and assign to the edge~$u_iv_j$ the colour~$\lambda(p_{i,j})$.
    The bipartite Ramsey theorem gives a monochromatic sub-biclique with~$t$ vertices on each side, say of colour~$c$.
    The corresponding points form a $t$-grid of colour~$c$.
\end{proof}

When the $x$- and $y$-coordinates of the grid are indexed by the same set (e.g.\ in an adjacency matrix),
they can be assumed to come from disjoint intervals at a slight cost:
\begin{lemma}\label{lem:grid-disjoint}
    Any $2t$-grid~$\Pc$ contains a $t$-subgrid $\Pc' \subset \Pc$
    such that the sets $X \eqdef \Xco(\Pc')$ and $Y \eqdef \Yco(\Pc')$ of $x$- and $y$-coordinates satisfy either $X < Y$ or $Y < X$.
\end{lemma}
\begin{proof}
    Let $(p_{i,j})_{i,j \in [2t]}$ be the points of~$\Pc$.
    There exists an~$x$ such that $\Xco(p_{i,j}) \le x$ whenever~$i \le t$, and $\Xco(p_{i,j}) > x$ whenever $i > t$:
    it suffices to take~$x$ to be the maximum of~$\Xco(p_{i,j})$ over all $i \le t$ and $j \in [2t]$.
    Similarly, there is a~$y$ such that $\Yco(p_{i,j}) \le y$ if and only if $j \le t$.
    If $x \le y$, then we choose~$\Pc'$ to consist of~$p_{i,j}$ for~$i \le t$ and $j > t$,
    which clearly satisfies the desired condition.
    Symmetrically, if $x > y$, we restrict the grid to $i > t$ and $j \le t$.
\end{proof}

The following technical lemma relates grids in incidence and adjacency matrices.
\begin{lemma}\label{lem:incidence-matrix-grid}
    Let $E \subseteq V^2$ be a binary relation such that $\adj(E,<_V)$ has no $k$-grid.
    Let~$<_\lex$ be the lexicographic ordering of~$E$, i.e.\ $(x,y) <_\lex (x',y')$ if either~$x <_V x'$, or $x = x'$ and $y <_V y'$.
    Then $\inc(E,<_V,<_\lex)$ has no $4k$-grid.
\end{lemma}
\begin{proof}
    For a directed edge $(x,y) \in E$, we write $s(e) = x$ and $t(e) = y$ for its source and target.
    Assume that the $Inc(E,<_V,<_\lex)$ contains an $\ell$-grid $(p_{i,j})_{i,j \in [\ell]}$.
    Each~$p_{i,j}$ is a~`1' in the matrix, corresponding to a vertex~$v_{i,j}$ and an incident edge~$e_{i,j}$. The definition of $\ell$-grid gives the following inequalities for any $i,i',j,j'$:
    \begin{align}
        \label{eq:ord-vtx} i<i' \quad & \text{implies} \quad v_{i,j} <_V v_{i',j'}, \text{ and}  \\
        \label{eq:ord-edge} j<j' \quad & \text{implies} \quad e_{i,j} <_\lex e_{i',j'}.
    \end{align}
    By definition of the lexicographic ordering, condition~\eqref{eq:ord-edge} also gives 
    \begin{equation}
        \label{eq:ord-source}
        j<j' \quad \text{implies} \quad s(e_{i,j}) \le_V s(e_{i',j'}).
    \end{equation}
    Say that the pair~$(i,j)$ is of \emph{source type} if $v_{i,j} = s(e_{i,j})$, and otherwise of \emph{target type}, meaning $v_{i,j} = t(e_{i,j})$.
    Observe that~$(i+1,j)$ and~$(i,j+1)$ cannot both be of source type, as this would imply both $v_{i,j+1} <_V v_{i+1,j}$ by~\eqref{eq:ord-vtx}, and $v_{i+1,j} \le_V v_{i,j+1}$ by~\eqref{eq:ord-source}.
    By grouping the pairs~$(i,j)$ into $2$-by-$2$ squares $(2i,2j),(2i,2j-1),(2i-1,2j),(2i-1,2j-1)$
    and keeping only one pair of target type in each, we obtain a new $\floor{\ell/2}$-grid in which all pairs~$(i,j)$ are of target type.
    Thus explicitly, we have edges~$f_{i,j}$ for~$i,j < \ell/2$ satisfying for any~$i,i',j,j'$ that
    \begin{align}
        \label{eq:ord-tgt} i<i' \quad & \text{implies} \quad t(f_{i,j}) <_V t(f_{i',j'}), \text{ and}  \\
        \label{eq:ord-edge2} j<j' \quad & \text{implies} \quad f_{i,j} <_\lex f_{i',j'}.
    \end{align}
    We claim that for~$j+1<j'$ and $i>1$, we also have $s(f_{i,j}) <_V s(f_{i',j'})$.
    Indeed, we have
    \[ f_{i,j} <_\lex f_{1,j+1} <_\lex f_{i',j'}, \]
    which translates into the corresponding non-strict inequalities on the respective source vertices.
    But on the other hand,~\eqref{eq:ord-tgt} gives $t(f_{i,j}) >_V t(f_{1,j+1})$.
    Then, for $f_{i,j} <_\lex f_{1,j+1}$ to hold, it must be that $s(f_{i,j}) <_V s(f_{1,j+1})$.
    Thus, when considering edges~$f_{i,j}$ only for~$i>1$ and odd values of~$j$,
    we have that the order of sources~$s(f_{i,j})$ coincides with the ordering of $j$-indices,
    and the ordering of targets~$t(f_{i,j})$ with that of $i$-indices.
    This yields an $\floor{\ell/4}$-grid in~$\adj(E,<_V)$, a contradiction.
\end{proof}

\subsection{Twin-width}\label{sec:prelim:tww}
We recall the definition of the twin-width~$\tww(G)$ of a graph~$G$ from~\cite{Bonnet2022twinwidth1}.

In a graph~$G$, two sets of vertices~$X,Y$ are \emph{non-homogeneous} if there exist both an edge and a non-edge between them.
Given a partition~$\Pc$ of the vertices of~$G$, the \emph{error graph} of~$\Pc$ is the graph
whose vertices are parts of~$\Pc$, and in which two parts are adjacent if they are non-homogeneous.
The \emph{error degree} of~$\Pc$ is the maximum degree of this error graph.

A \emph{contraction sequence} for~$G$ is a sequence $\Pc_n,\dots,\Pc_1$ of partitions of~$V(G)$,
where~$\Pc_n$ is the partition into singletons, $\Pc_1$ is the partition in a single part,
and~$\Pc_{i-1}$ is obtained by merging two parts of~$\Pc_i$.
Its \emph{width} is the maximum error degree of $\Pc_n,\dots,\Pc_1$.
Finally, the \emph{twin-width}~$\tww(G)$ is the minimum width of a contraction sequence for~$G$.

More generally, in a binary relational structure $S = (V,E_1,\dots,E_\ell)$,
two sets of vertices~$X,Y$ are non-homogeneous if there is some relation~$E_i$,
and some vertices $x_1,x_2 \in X$, $y_1,y_2 \in Y$ satisfying $x_1y_1 \in E_i$ and $x_2y_2 \not\in E_i$,
or symmetrically $y_1x_1 \in E_i$ and $y_2x_2 \not\in E_i$.
For instance, in an ordered graph~$(V,E,<)$, two sets~$X,Y$ are homogeneous if and only if
(1) they are homogeneous for~$E$ (i.e.\ fully adjacent or fully non-adjacent), and (2) either $X < Y$ or $Y < X$.
Then, the error graph of a partition~$\Pc$ is again defined as having edges between all pairs of non-homogeneous parts,
allowing to define twin-width for binary structures.

The fundamental characterisation of twin-width is that it corresponds to ordered graphs whose adjacency matrix excludes some high-rank grid-like structures~\cite[Theorem 7]{twin-width4}.
This characterisation is furthermore efficient.
These high-rank grids contain our notion of grid, thus this characterisation implies the following.
\begin{theorem}[\cite{Bonnet2022twinwidth1,twin-width4}]\label{thm:grid-tww}
    There are functions~$f,g$ and an algorithm which, given~$(V,E,<)$ an ordered graph
    whose matrix~$\adj(E,<)$ has no $k$-grid, computes a contraction sequence witnessing that $\tww(V,E,<) \le f(k)$ in time $g(k) \cdot |V|^{O(1)}$.
\end{theorem}
This characterisation also implies that ordered graphs of high twin-width have a Ramsey-like property:
\begin{theorem}[\cite{twin-width4}]\label{thm:ordered-tww-ramsey}
    There is a function~$R(k,\ell)$ such that if $(V,E,<)$ is an ordered graph with $E = E_1 \cup \dots \cup E_m$, and $\tww(V,E_i,<) \le k$ for all~$i$, then $\tww(V,E,<) \le R(k,\ell)$.
\end{theorem}
\begin{proof}
    This follows from the characterisation of bounded twin-width for ordered graphs by excluding so-called high-rank divisions \cite[Theorems~7 and~23]{twin-width4}, and the Ramsey property for these high-rank divisions \cite[Lemma~24]{twin-width4}.
\end{proof}

It is easy to observe that the twin-width of a graph is the maximum twin-width of its connected components.
We will use the following variant of this result for ordered graphs.
\begin{lemma}\label{lem:ordered-components}
    Let~$G = (V,E,<)$ be an ordered graph, with~$V$ partitioned into intervals $A_1 < \dots < A_m < B_1 < \dots < B_m$, such that
    all edges are between~$A_i$ and~$B_i$ for some~$i$, and the subgraphs $(A_i \cup B_i, E, <)$ have twin-width at most~$k$ for all~$i$.
    Then $\tww(G) \le 2k+2$.
\end{lemma}
\begin{proof}
    Call~$G_i$ the ordered subgraph induced by $A_i \cup B_i$, satisfying $\tww(G_i) \le k$ by assumption.
    By \cite[Theorem~4.1]{Bonnet2022twinwidth1}, there is a contraction sequence for~$G_i$ of width at most~$2k+2$
    in which vertices of~$A_i$ are never contracted with those of~$B_i$
    (such a contraction sequence ends with the bipartition $\{A_i,B_i\}$, instead of the trivial partition $\{A_i \cup B_i\}$ as it normally should).
    This contraction sequence can be replicated in~$G$,
    and since~$A_i,B_i$ are intervals in~$<$ and have no edge to the rest of the graph, this does not result in any additional edge in the error graph.

    Applying this argument for each~$i$, we obtain a partial contraction sequence of width at most~$2k+2$, ending with the partition $\{A_1,\dots,A_m,B_1,\dots,B_m\}$.
    One can complete this contraction sequence by merging the parts~$A_1,A_2$ immediately followed by~$B_1,B_2$, and repeating similarly until only two parts $\bigcup_i A_i$ and $\bigcup_i B_i$ remain, which can be merged together.
    Throughout this second phase of the merge sequence, the error degree never exceeds~2.
\end{proof}

\subsection{First-order logic}\label{sec:prelim:FO}
We present first-order (FO) logic and transductions restricted to \emph{binary} relational structures, as this is sufficient for our purposes.
A \emph{relational signature} is a set $\Sigma = \{E_1,\dots,E_\ell\}$ of \emph{relation symbols}.
A \emph{$\Sigma$-structure}~$S$ is a relational structure over this signature,
i.e.\ a structure with vertex set~$V(S)$ and an edge set~$E_i(S)$ for each symbol~$E_i \in \Sigma$.

For a signature~$\Sigma$, one defines \emph{FO formulas}~$\phi$ over~$\Sigma$.
They consist of quantifiers over vertices $\forall x,\phi$, $\exists x,\phi$, boolean operators~$\lor,\land,\lnot$,
predicates~$E_i(x,y)$ testing whether~$xy$ is an edge in~$E_i$, as well as the equality predicate~$x=y$.
For example, the formula
\[ \phi(x,y) = E(x,y) \lor \exists z.\ (E(x,z) \land E(z,y)) \]
expresses that~$x,y$ are at distance at most~two.
When the formula~$\phi$ has free variables, they are written as arguments of~$\phi$ as above.
A formula without free variables is called a \emph{sentence}.
If $\phi(x_1,\dots,x_r)$ is a formula over~$\Sigma$, $S$ is a $\Sigma$-structure, and $v_1,\dots,v_r \in V(S)$ are vertices,
then $S \models \phi(v_1,\dots,v_r)$ denotes that~$\phi$ is \emph{satisfied} by~$S$, with the variables $x_1,\dots,x_r$ interpreted as $v_1,\dots,v_r$ respectively.

The FPT model checking algorithm for twin-width shows the following.
\begin{theorem}[{\cite[Theorem~1]{Bonnet2022twinwidth1}}]\label{thm:model-checking}
    Given a graph~$G$, a contraction sequence of width~$k$ for~$G$, and a sentence~$\phi$,
    one can test whether $G \models \phi$ in time $f(k,|\phi|) \cdot |V(G)|$.
\end{theorem}

Given two signatures~$\Sigma,\Gamma$, an \emph{FO interpretation}~$\Phi$ from~$\Sigma$ to~$\Gamma$
is a map, defined in FO logic, from $\Sigma$-structures to $\Gamma$-structures.
Precisely, $\Phi$ is described by giving
\begin{enumerate}
    \item for each symbol~$E_i \in \Gamma$, an FO formula~$\phi_{E_i}(x,y)$ over the language of~$\Sigma$, and
    \item one last FO formula~$\phi_V(x)$ again over~$\Sigma$.
\end{enumerate} 
Given a $\Sigma$-structure~$S$, its image~$\Phi(S)$ is the relational structure where the vertex set is restricted to vertices satisfying~$\phi_V$, i.e.\
\[ V(\Phi(S)) = \{x \in V(S) \ : \ S \models \phi_V(x) \}, \]
and for each symbol~$E_i \in \Gamma$, the edge set~$E_i(\Phi(S))$ is described by~$\phi_{E_i}$, i.e.\
\[ E_i(\Phi(S)) = \{xy \ : \ x,y \in V(\Phi(S)) \text{ and } S \models \phi_{E_i}(x,y) \}. \]

\emph{Transductions} are a non-deterministic generalisation of interpretations.
That is, a~$\Sigma$-to-$\Gamma$ transduction outputs a set of $\Gamma$-structures, rather than a single one.
Let~$\Sigma$ be a signature, and $U_1,\dots,U_k$ be~$k$ unary relation symbols (i.e.\ of arity~1), disjoint from~$\Sigma$.
The $k$-colouring is the one-to-many operation which maps a $\Sigma$-structure $S$
to all possible extensions of~$S$ as $(\Sigma \uplus \{U_1,\dots,U_k\})$-structures~$S^+$,
meaning that~$V(S) = V(S^+)$ and~$R(S) = R(S^+)$ for any~$R \in \Sigma$,
while the~$U_i(S^+)$ are chosen as arbitrary subsets of~$V(S^+)$.
An \emph{FO transduction} is the composition of a $k$-colouring (for fixed~$k$), followed by an FO interpretation.

It is folklore that interpretations and transductions can be composed.
\begin{lemma}
    \label{lem:transduction-comp}
    If~$\Phi,\Psi$ are FO transductions (resp.\ interpretations) from~$\Sigma$ to~$\Gamma$ and~$\Gamma$ to~$\Delta$ respectively,
    then the composition~$\Psi \circ \Phi$ is an FO transduction (interpretation) from~$\Sigma$ to~$\Delta$.
\end{lemma}

Transductions preserve bounded twin-width, and this can be made efficient.
\begin{theorem}[{\cite[Theorem~8.1]{Bonnet2022twinwidth1}}]\label{thm:tww-transduction}
    For any FO transduction~$\Phi$, there is a function~$f$ such that
    for any binary structure~$S$ and $T \in \Phi(S)$, $\tww(T) \le f(\tww(S))$.

    Furthermore, given a contraction sequence of width~$k$ for~$S$, and the colouring of~$S$ used to obtain~$T$ through~$\Phi$, one can compute in time $g(k) \cdot |S|^{O(1)}$ a contraction sequence of width~$f(k)$ for~$T$ for some function~$g$ depending on~$\Phi$.
\end{theorem}

A graph class~$\Cc$ is called \emph{monadically independent} if it FO transduces the class of all graphs, i.e.\ if there exists a transduction~$\Phi$ such that~$\Phi(\Cc)$ contains all graphs.
If there is no such transduction, then~$\Cc$ is called \emph{monadically dependent}.
\Cref{thm:tww-transduction} implies that bounded twin-width classes are monadically dependent.

\begin{theorem}[\cite{flip-breakability}]\label{thm:independent-hardness}
    If~$\Cc$ is a hereditary monadically independent class of graphs,
    then FO model checking in~$\Cc$ is \aw-hard.
\end{theorem}

To prove that a class~$\Cc$ is monadically independent, we will use the next construction, proposed in~\cite{bonnet2022delineation}.
Our definition differs slightly, but the idea and purpose are the same.
\begin{definition}\label{def:transversal-pair}
A \emph{transversal pair} $T_k$ consists of three sets $A = \{a_i : i \in [k]\}$, $B = \{b_{i, j} : i, j \in [k]\}$ and $C = \{c_j: j \in [k]\}$ 
with edges $a_i b_{i',j'}$ if and only if $i \le i'$, and $b_{i,j} c_{j'}$ if and only if $j \le j'$ (see \cref{fig:transversal-pair}).
Edges within~$A,B,C$ or between~$A$ and~$C$ are unconstrained.
\end{definition}

\begin{figure}[htb]
    \centering
    \begin{tikzpicture}
        \tikzstyle{vertex}=[black,fill,draw,circle,inner sep=0pt, minimum width=1.2mm]
        \def\n{3}

        \foreach \i in {1,...,\n}{
            \node[vertex, label=left:$a_\i$] (a\i) at (0,-\i) {};
            \node[vertex, label=above:$c_\i$] (c\i) at (\i,0) {};
            \foreach \j in {1,...,\n}{
                \node[vertex] (b\i\j) at (\j,-\i) {};
            }
        }
        \node[label=below:$b_{3,1}$] at (b31) {};
        \node[label=below:$b_{3,2}$] at (b32) {};
        \node[label=below right:$b_{3,3}$] at (b33) {};
        \node[label=right:$b_{1,3}$] at (b13) {};
        \node[label=right:$b_{2,3}$] at (b23) {};

        \foreach \i in {1,...,\n}{
            \foreach \j in {1,...,\n}{
                \foreach \k in {1,...,\n}{
                    \pgfmathsetmacro{\bend}{\k*3}
                    \ifnum \i>\j \else
                        \draw (a\i) edge [bend left=\bend] (b\j\k);
                        \draw (c\j) edge [bend left=\bend] (b\k\i);
                    \fi
                }
            }
        }

    \end{tikzpicture}
    \caption{A transversal pair $T_3$.}
    \label{fig:transversal-pair}
\end{figure}
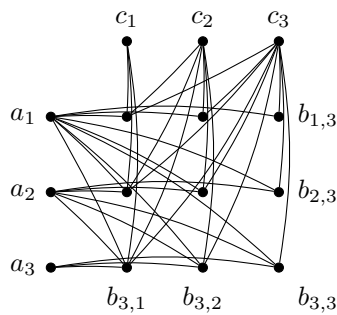

\begin{lemma}\label{lem:transversal-pair-NIP}
    Let~$\Cc$ be a graph class such that for any~$k$, some $G \in \Cc$ contains a transversal pair~$T_k$. Then $\Cc$ is monadically independent.
\end{lemma}
This is folklore, and a very special case of results of~\cite{flip-breakability}.
We give a proof for completeness.%
\begin{proof}
    We show that the class of all bipartite graphs can be transduced from the class of all transversal pairs.
    This implies that all graphs can be transduced from transversal pairs too,
    since any graph can be transduced from its 1-subdivision, which is bipartite.

    Recall that~$T_n$ consists of three vertex sets $A = \{a_i\}_{i \in [n]}$, $B = \{b_{i,j}\}_{i,j \in [n]}$ and $C = \{c_j\}_{j \in [n]}$,
    with edges $a_ib_{i',j'}$ whenever $i \le i'$, and $b_{i,j}c_{j'}$ whenever $j \le j'$.
    We first want to modify the edges, so that~$b_{i,j}$ is adjacent only to~$a_i$ and~$c_j$.
    The resulting graph will exactly be the biclique~$K_{n,n}$ with each edge subdivided once.

    To this end, observe that for any~$i,i',j'$, we have $i=i'$ if and only if (1) $a_i$ is adjacent to~$b_{i',j'}$,
    and (2) $a_i$ is minimal with this property, in the sense that for any $a' \in A$ adjacent to $b_{i',j'}$, $N(a_i) \subseteq N(a')$.
    This can be expressed as an FO transduction:
    after marking the sets~$A$, $B$, and~$C$ using non-deterministic colouring, the two conditions are easily expressed by an FO formula.
    The same technique can be applied between~$B$ and~$C$,
    and we thus obtain an FO transduction~$\Phi_1$ such that $\Phi_1(T_n)$ contains the biclique~$K_{n,n}$ subdivided once.

    Now consider~$H$ a bipartite graph with partition $(X,Y)$ with $|X| = x$ and $|Y| = y$.
    We take the subdivided~$K_{n,n}$ obtained from the previous step $n = \max\{x, y\}$.
    Identify~$X$ with vertices $\{a_1,\dots,a_x\}$ from~$A$, and~$Y$ with $\{c_1,\dots,c_y\}$ from~$B$.
    Also define the set $F = \{b_{i,j} : a_ic_j \in E(H)\}$ of vertices in~$B$ corresponding to edges of~$H$.
    The following transduction then produces~$H$ from the subdivision of~$K_{n,n}$:
    mark the sets~$X$, $Y$, and~$F$ with non-deterministic colouring,
    and define the edge relation of~$H$ by the formula
    \[ \varphi(u, v) \equiv X(u) \land Y(v) \land \exists w (F(w) \land E(u, w) \land E(w, v)). \]
    Finally delete vertices outside~$X$ and~$Y$ to obtain~$H$.

    Thus from the class of all transversal pair, we can by three consecutive transduction obtain all subdivided bicliques, then all bipartite graphs, and finally all graphs.
    As the composition of two transductions is itself a transduction (\cref{lem:transduction-comp}), this concludes the proof.
\end{proof}

\subsection{Merge-width}
Merge-width~\cite{merge-width} is a major generalisation of twin-width, designed to also capture graphs with bounded degree, or more generally bounded expansion.
The definition also involves sequences of vertex partitions, but the notion of error is finer:
one counts errors for each vertex individually, rather than for each part.
Merge-width really is a family of parameters: one defines \emph{merge-width at radius~$r$}, denoted~$\mw_r(G)$, for any $r \in \N$.

A \emph{construction sequence} is a sequence $(\Pc_1,E_1,N_1) , \ldots, (\Pc_m,E_m,N_m)$,
where~$\Pc_i$ are partitions of the vertex set~$V$, and~$E_i,N_i$ are disjoint sets of \emph{resolved} edges and non-edges, obtained as follows.
\begin{itemize}
    \item The sequence starts with the partition into singletons $\Pc_1 = \{\{v\} : v \in V\}$, and with no resolved edge or non-edge, i.e.\ $E_1 = N_1 = \emptyset$.
    \item One advances from step~$i$ to~$i+1$ with one of the following three operations:
    \begin{enumerate}
        \item merge two parts of~$\Pc_i$ to obtain~$\Pc_{i+1}$,
        \item \emph{positively resolve} two parts $X,Y \in \Pc_i$ (possibly $X=Y$), i.e.\ define~$E_{i+1}$ by adding to~$E_i$ all pairs $x \in X$, $y \in Y$ that are not yet in~$E_i$ or~$N_i$,
        \item or \emph{negatively resolve}~$X,Y$, adding these pairs to~$N_{i+1}$ instead.
    \end{enumerate}
    \item The sequence ends when all vertices are merged into $\Pc_n = \{V\}$, and all pairs are resolved in either~$E_m$ or~$N_m$. The resulting graph is~$(V,E_m)$.
\end{itemize}
Define $R_i \eqdef E_i \cup N_i$ the set of resolved pairs at time~$i$.
Then for $r \in \N \cup \{\infty\}$, the \emph{radius-$r$ width of $\Pc$} is
\[
    \max\limits_{t \in [m]} \, \max\limits_{v \in V} \Big| \bigl\{A \in \Pc_t : A \text{ is at distance at most~$r$ from } v \text{ in } (V, R_t) \bigr\} \Big|.
\]
The \emph{radius-$r$ merge-width of $G$}, denoted by $\mw_r(G)$, is the least radius-$r$ width of a construction sequence creating~$G$.
A graph class $\Cc$ is said to have \emph{bounded merge-width} if there is a function~$f$ such that
for all $r \in \N$ and $G \in \Cc$, it holds that $mw_r(G) \le f(r)$.

The main motivation for merge-width, and the reason we use it, is the following generalisation of \cref{thm:model-checking}.
\begin{theorem}[\cite{merge-width}]\label{thm:mw-model-checking}
    If~$\Cc$ is a class with bounded merge-width, then~$\Cc$ is monadically dependent, and FO model checking in~$\Cc$ is FPT.\footnote{%
        To be accurate, if radius-$r$ merge-width of graphs in~$\Cc$ is bounded by some \emph{constructible} function of~$r$,
        then model checking is truely FPT.
        Otherwise, one only obtains a \emph{non-uniform} FPT algorithm.
        In our uses of \cref{thm:mw-model-checking}, the merge-width bounding function will simply be~$2^{2^{O(r)}}$, so this is not a concern.
    }
\end{theorem}

\subsection{Geometric graphs}\label{sec:prelim:geom}
For a set family~$\Fc$, the \emph{intersection graph}~$G_\Fc$ has~$\Fc$ for vertex set, with an edge between $A,B \in \Fc$ if and only if~$A$ and~$B$ intersect. The family~$\Fc$ is called a \emph{representation} of~$G_\Fc$.

\paragraph*{Circular arc graph}
A \emph{circular arc graph} is an intersection graph of arcs around a circle.
We work with a discrete representation in $[n] \eqdef \{1,...,n\}$ modulo~$n$:
a (discrete) \emph{circular arc} is an interval~$[i,j]$ modulo~$n$;
explicitly, $[i,j] \eqdef \{k: i \le k \le j\}$ when $i \le j$, and $[i,j] \eqdef \{k : k \ge i \text{ or } k \le j\}$ when $i > j$.
A graph~$G$ is a circular arc graph if and only if it is the intersection graph of some family~$\Fc$ of (discrete) circular arcs as above for some~$n$.
We frequently identify the vertices of~$G$ with the arcs of~$\Fc$.
When $I = [i,j]$ is a circular arc, we denote by $\Lend(I) \eqdef i$ and $\Rend(I) \eqdef j$ its endpoints.

\begin{theorem}[\cite{tucker1980circular}]\label{thm:cag-detection}
    Given a graph~$G$, one can compute a circular arc representation of~$G$ or detect that it is not a circular arc graph in polynomial time.
\end{theorem}

\paragraph*{Axis-parallel unit segment graphs}
An axis-parallel unit segment (APUS) is a segment of length~1 in the plane that is horizontal or vertical.
We denote by~$\Hseg(x,y)$ the horizontal segment from~$(x,y)$ to~$(x+1,y)$,
and by~$\Vseg(x,y)$ the vertical one from~$(x,y)$ to~$(x,y+1)$.
An APUS graph is the intersection graph of a family of APUS.
We again work with discrete representations:
Any APUS graph~$G$ has a representation~$\Fc$ such that there is some~$\eta(\Fc)$ with $1/\eta(\Fc) \in \N$ such that coordinates of segments in~$\Fc$ are non-negative multiples of~$\eta(\Fc)$.

A family~$\Fc$ of APUSs is \emph{non-degenerate} if no two horizontal (resp.\ two vertical) segments of~$\Fc$ intersect.
We say~$G$ is a non-degenerate APUS graph if it is the intersection graph of a non-degenerate APUS family.
This is equivalent to~$G$ being the intersection graph of an arbitrary family~$\Fc$ of APUS, ignoring intersections between horizontal (vertical) segments (i.e.~$G$ has no edge for such intersections).
Non-degenerate APUS graphs are bipartite, with horizontal segments on one side and vertical ones on the other.

We will use APUS representations with minimized coordinates, which is harder for non-degenerate APUS families.
Thus we prefer the second point of view on non-degenerate APUS graphs:
degenerate intersections can exist, but are ignored in intersection graphs.

\paragraph*{Terrain Visibility Graphs}
A \emph{1.5D terrain} is a polygonal curve $P = p_1,\dots,p_n$ with increasing $x$-coordinates $\Xco(p_1) < \dots < \Xco(p_n)$.
The \emph{visibility graph} of~$P$ has vertex set~$P$, with an edge between~$p_i,p_j$ whenever they see each other,
if we think of the area below~$P$ as obstructed.
Formally, $p_ip_j$ is \emph{not} an edge if and only if~$p_k$ is above the line from~$p_i$ to~$p_j$ for some $i<k<j$.
As such, two adjacent vertices in the terrain will always be visible to each other.

\section{Delineation of Circular Arc Graphs}
\label{sec:cag}
In this section, we prove \cref{thm:main} for circular arc graphs.

Given a family of circular arcs~$\Fc$, we define its \emph{endpoints matrix} $M_\Fc$, where rows correspond to the right endpoints and columns to the left endpoints of the arcs. For each entry at $(i, j)$ in $M_\Fc$, there is a $1$ if and only if $[i, j] \in \Fc$.

\begin{lemma}\label{lem:cag-gridthm}
    There are functions~$f,g$ such that if~$M_\Fc$ has no $k$-grid,
    then one can compute a contraction sequence witnessing that $\tww(G_\Fc) \le f(k)$ in FPT time $g(k) \cdot |\Fc|^{O(1)}$.
\end{lemma}
\begin{proof}
    Assume that~$\Fc$ is represented within~$[n]$.
    We identify intervals $[x,y] \in \Fc$ with pairs $(x,y)$, so that~$\Fc$ describes a binary relation on~$[n]$.
    Its adjacency matrix $\adj(\Fc,<)$ with the natural ordering of~$[n]$ is exactly~$M_\Fc$.
    It has no $k$-grid, hence by \cref{lem:incidence-matrix-grid}, the incidence matrix $\inc(\Fc,<,<_\lex)$ has no $4k$-grid, where~$<_\lex$ is the lexicographic ordering of~$[n]^2$.

    Construct now the following ordered graph $H = (V,E,<_V)$.
    The vertex set is $V \eqdef [n] \uplus \Fc$, ordered with~$[n]$ before~$\Fc$, with the natural ordering inside~$[n]$, and the lexicographic ordering inside~$\Fc$.
    In~$E$, each interval $[x,y] \in \Fc$ is connected to its two endpoints.
    Thus, the adjacency matrix of~$H$ consists of $\inc(\Fc,<,<_\lex)$ and its transpose, arranged as diagonal blocks.
    This matrix has no $(4k+1)$-grid.
    By \cref{thm:grid-tww}, one may compute in FPT time a contraction sequence of bounded width for~$(V,E,<_V)$.

    To conclude using \cref{thm:tww-transduction}, it suffices to show that there is an FO transduction~$\Phi$ such that $G_\Fc \in \Phi(H)$. Given~$H$, the transduction~$\Phi$ proceeds as follows.
    \begin{enumerate}
        \item Firstly, non-deterministic colouring of~$V$ is used to distinguish~$[n]$ and~$\Fc$,
        and also to indicate for each interval $[x,y] \in \Fc$ whether or not $x \le y$.

        \item Using these colours, we can write a formula $\phi(v,e)$ testing if $v \in [n]$ belongs to the interval $e\in \Fc$.
        Indeed, let $x,y \in [n]$ be the neighbours of~$e$, with~$x<y$.
        Then the interval~$e$ is either~$[x,y]$ or~$[y,x]$, indicated by the colour given to~$e$ at step~1.
        It suffices for~$\phi(v,e)$ to check that $x \le v \le y$ in the first case, and that $v \le y \lor x \le v$ in the second.
        
        As an extreme case, one could have $e = [x,x]$, which is adjacent to~$x$ only.
        Then~$x$ is the only point contained in~$e$.
        All of the above is easily expressed by a first-order formula.
        
        \item Now, one can test if~$e,f \in \Fc$ intersect with the formula $\psi(e,f) = \exists v.\ \phi(e,v) \land \phi(f,v)$.

        \item Finally, vertices of~$[n]$ are deleted, keeping only~$\Fc$ and the edges defined by~$\psi$. \qedhere
    \end{enumerate}
\end{proof}

We say a circular arc representation $\Fc$ is \emph{minimized} if its circular arcs are inclusion-wise minimal while representing~$G_\Fc$, i.e.\ for any~$I \in \Fc$, removing either endpoint of~$I$ would cause it to no longer intersect some other arc $I' \in \Fc$.
Clearly, a circular arc representation can be minimized in polynomial time.
We will use the following property.
\begin{lemma}\label{lem:min-arc-rep}
    If~$\Fc$ is a minimized circular arc representation and $I = [i,j] \in \Fc$ with~$i \neq j$, then there exist~$I_L, I_R \in \Fc$ such that $I \cap I_L = \{i\}$ and $I \cap I_R = \{j\}$.
\end{lemma}
\begin{proof}
    If~$I$ violates the condition, then it can be shortened by deleting either~$i$ or~$j$.
\end{proof}

\begin{lemma}\label{lem:cag-trans-pair}
    Let $k$ be a positive integer. Consider a minimized circular arc representation~$\Fc$ of a graph~$G_\Fc$.
    If~$M_\Fc$ contains a $(4k+2)$-grid, then~$G_\Fc$ contains a transversal pair~$T_k$.
\end{lemma}
\begin{proof}
    Consider a $(4k+2)$-grid in~$M_\Fc$. It consists of intervals~$I_{i,j}$ satisfying $\Lend(I_{i,j}) < \Lend(I_{i',j'})$ whenever~$i<i'$, and $\Rend(I_{i,j}) < \Rend(I_{i',j'})$ whenever~$j<j'$.
    By \cref{lem:grid-disjoint}, a $(2k+1)$-grid can be extracted so that the $x$- and $y$-coordinates belong to disjoint intervals.
    Thus we assume that $i,j \in [2k+1]$, and that $\Lend(I_{i,j}) < \Rend(I_{i',j'})$ for all~$i,i',j,j'$, the case $\Lend(I_{i,j}) > \Rend(I_{i',j'})$ being similar.

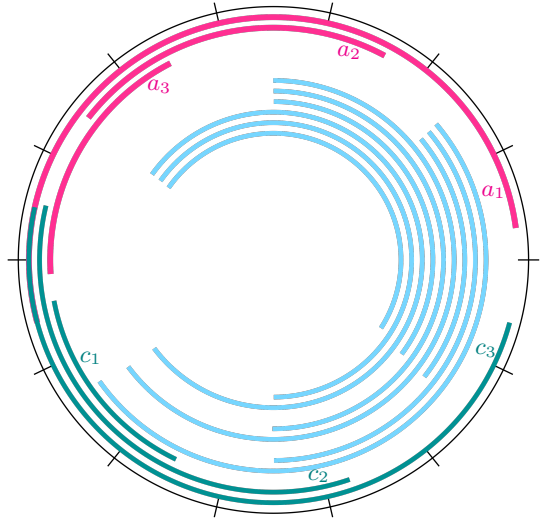
\begin{figure}
    \centering
    \begin{tikzpicture}
        \node at (0,0) {\includegraphics[width=7cm]{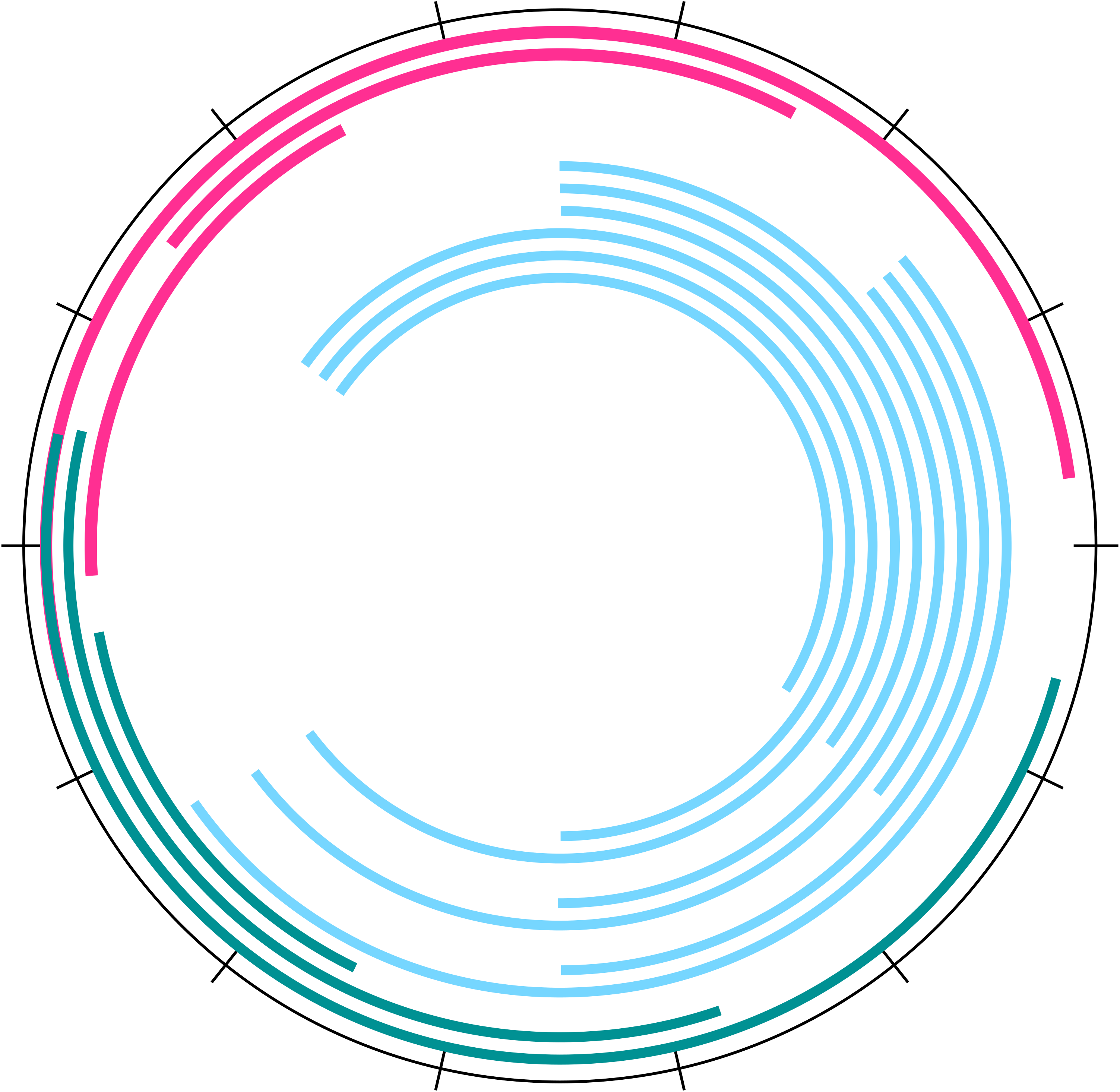}};
        
        \node at (2.9,0.9) {\color{magenta} \small $a_1$};
        \node at (1.0,2.75) {\color{magenta} \small $a_2$};
        \node at (-1.5,2.3) {\color{magenta}\small $a_3$};

        \node at (2.8,-1.2) {\color{teal} \small $c_3$};
        \node at (0.6,-2.865) {\color{teal} \small $c_2$};
        \node at (-2.4,-1.3) {\color{teal} \small $c_1$};
    \end{tikzpicture}
    \caption{A circular arc representation of a transversal pair $T_3$ with the tripartition $(A,B,C)$. \textcolor{magenta}{Magenta} circular arcs represent vertices in $A$, \textcolor{cyan}{cyan} circular arcs in $B$, and \textcolor{teal}{teal} circular arcs in $C$.}
    \label{fig:circular_arc_transversal_pair}
\end{figure}

    Let us now start constructing the desired transversal pair~$T_k$.
    The central vertices $b_{i,j} \eqdef I_{2i,2j}$ of the transversal pair are obtained by picking every other row and column of the grid.
    For the left vertices~$a_i$, we first extract auxiliary intervals $x_i \eqdef I_{2i+1,2k+1}$ for $i \in [k]$.
    We claim that for all~$i,i',j' \in [k]$
    \begin{align}
        & \text{if $i \ge i'$, then $\Lend(x_i) \in b_{i',j'}$, and} \label{eq:case-intersect} \\
        & \text{if $i < i'$, then $b_{i',j'} \subseteq x_i \setminus \{\Lend(x_i)\}$.} \label{eq:case-disjoint}
    \end{align}
    Indeed, if $i \ge i'$, then $2i+1 > 2i'$, which gives
    \begin{equation}
        \Lend(b_{i',j'}) = \Lend(I_{2i',2j'}) < \Lend(I_{2i+1,2k+1}) = \Lend(x_i).
    \end{equation}
    In addition, since all left endpoints are smaller than all right endpoints for intervals of the grid, we also have $\Lend(x_i) < \Rend(b_{i',j'})$. Thus~$b_{i',j'}$ contains~$\Lend(x_i)$, proving~\eqref{eq:case-intersect}.
    On the other hand, if~$i < i'$, then we have $2i+1 < 2i'$ and $2j' < 2k+1$, which gives 
    \begin{equation}
        b_{i',j'} = I_{2i',2j'} \subset I_{2i+1,2k+1} = x_i,
    \end{equation}
    the inclusion being strict at both endpoints, which implies~\eqref{eq:case-disjoint}.
    
    Using the minimality of $\mathcal{F}$ and \cref{lem:min-arc-rep}, there is an interval $a_i \in \mathcal{F}$ such that $x_i \cap a_i = \{\Lend(x_i)\}$.
    It follows from~\labelcref{eq:case-intersect,eq:case-disjoint} that~$a_i$ and~$b_{i',j'}$ intersect if and only if~$i \ge i'$.
    Symmetrically, we define $y_j = I_{1,2j-1}$, use \cref{lem:min-arc-rep} to find~$c_j$ such that $y_j \cap c_j = \{\Rend(y_j)\}$, and obtain that~$c_j$ and~$b_{i',j'}$ intersect if and only if~$j \le j'$.
    Up to reversing the order of indices, the former description of $a_i,b_{i,j},c_j$ gives a transversal pair. 
\end{proof}

Our main theorem for circular arc graphs directly follows from \cref{lem:cag-gridthm,lem:cag-trans-pair}.
\begin{theorem} \label{thm:delineation_circular_arc}
    Let~$\Cc$ be a subclass of circular arc graphs.
    Then~$\Cc$ has bounded twin-width if and only if it is monadically dependent.
    Further, when~$\Cc$ is hereditary, and assuming $\fpt \neq \aw$, these are also equivalent to
    FO model checking in~$\Cc$ being FPT.
\end{theorem}
\begin{proof}
    For each $G \in \Cc$, call~$\Fc_G$ the circular arc representation computed by \cref{thm:cag-detection}, after minimizing it.
    Consider the endpoint matrices $\Mc = \{M_{\Fc_G} : G \in \Cc\}$.
    If matrices in~$\Mc$ have no $k$-grid for some~$k \in \N$, then by \cref{lem:cag-gridthm} contraction sequences of width~$f(k)$ can be computed in FPT time for graphs in~$\Cc$.
    Thus~$\Cc$ has bounded twin-width (and thus is monadically dependent), and model checking in~$\Cc$ is FPT by \cref{thm:model-checking}.

    If however matrices in~$\Mc$ have arbitrarily large grids, then by \cref{lem:cag-trans-pair}, $\Cc$ contains arbitrarily large transversal pairs.
    Thus~$\Cc$ is monadically independent by \cref{lem:transversal-pair-NIP}.
    If~$\Cc$ is hereditary, this implies by \cref{thm:independent-hardness} that model checking in~$\Cc$ is \aw-hard.
\end{proof}

\section{Delineation of Axis-Parallel Unit Segment Graphs}
\label{sec:apus}
This section proves \cref{thm:main} for APUS graphs.
Recall that we work with non-degenerate APUS graphs: intersections between two horizontal (resp.\ vertical) segments are ignored.

As in the circular arc case, to a family~$\Fc$ of APUS, we associate the \emph{endpoint matrix}~$M_\Fc$:
the columns (resp.\ rows) of are the $x$- ($y$-) coordinates of all endpoints of segments in~$M_\Fc$,
and the entry at position~$(x,y)$ is~1 if and only $\Hseg(x,y) \in \Fc$ or $\Vseg(x,y) \in \Fc$.
A variant of \cref{lem:cag-gridthm} shows that when~$M_\Fc$ has no $k$-grid, then~$G_\Fc$ has bounded twin-width.
However, the converse fails: even with appropriate minimality conditions on~$\Fc$, a grid in~$M_\Fc$ does not imply that~$G_\Fc$ has unbounded twin-width.
This e.g.\ fails if the grid is formed by segments that are pairwise at distance more than~1 from each other.
The solution, inspired by \cite[Lemma~56]{bonnet2022delineation}, is to partition the APUS graph into regions corresponding to unit squares, and look for a grid in each region independently.

\subsection{Splitting along unit squares}
\begin{figure}
    \centering
    \includegraphics[width=6cm]{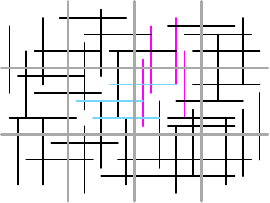}
    \caption{An instance being split along unit squares (drawn in {\color{gray}gray}).
    The {\color{cyan}cyan} segments are the set $\Hc_{s-1,t}$, while the {\color{magenta}magenta} segments are the set $\Vc_{s,t}$. Together, they form $\Fc_{s,t,\ulc}$.
    \label{fig:unit-grid}}
\end{figure}
In an APUS family~$\Fc$, for~$i,j \in \Z$, denote by~$\Hc_{s,t}$ the subset of segments $\Hseg(x,y)$ where $x \in [s,s+1)$ and $y \in [t,t+1)$, and define~$\Vc_{s,t}$ similarly for vertical segments.
Clearly the sets~$\Hc_{s,t},\Vc_{s,t}$ for $s,t \in \Z$ partition~$\Fc$.
Also, observe that segments in~$\Hc_{s,t}$ and~$\Vc_{s',t'}$ can intersect only if $s' \in \{s,s+1\}$ and $t' \in \{t-1,t\}$.
Let $\Fc_{s,t,\urc}$ denote the set $\Hc_{s,t} \cup \Vc_{s,t}$, and $E_{s,t,\urc}$ the edges of~$G_\Fc$ contained in~$\Fc_{s,t,\urc}$.
In other words, edges in~$E_{s,t,\urc}$ come from segments intersecting in the unit square $[s,s+1) \times [t,t+1)$,
for segments coming from the right and top sides of this square (hence~$\urc$ in the notation).
Similarly, define $\Fc_{s,t,\ulc} = \Hc_{s-1,t} \cup \Vc_{s,t}$,  $\Fc_{s,t,\lrc} = \Hc_{s,t} \cup \Vc_{s,t-1}$, and  $\Fc_{s,t,\llc} = \Hc_{s-1,t} \cup \Vc_{s,t-1}$, as well as $E_{s,t,\ulc},E_{s,t,\lrc},E_{s,t,\llc}$ the corresponding edge sets.
The sets~$E_{s,t,d}$ for $s,t \in \Z$ and $d \in \{\urc,\ulc,\lrc,\llc\}$ partition the edge set of~$G_\Fc$.
We work on each~$E_{s,t,d}$ independently with the techniques of \cref{sec:cag},
and show that when all have bounded twin-width, then so does~$G_\Fc$. Let us start with the latter.

In general, edge partitions do not preserve bounded twin-width, but they do for ordered graphs (\cref{thm:ordered-tww-ramsey}).
To use this result, we define an ordering~$<_\seg$ of segments.

A natural attempt is to take the lexicographic ordering~$<_\lex$ of segments,
i.e.\ $\Hseg(x,y) <_\lex \Hseg(x',y')$ if and only if $x<x'$, or $x=x'$ and $y<y'$, and similarly for vertical segments.
However this works poorly with twin-width:
imagine segments $\Fc = \{h_1,\dots,h_n,v_1,\dots,v_n\}$ such that $h_i,v_i$ intersect close to~$(0,2i)$.
In particular, $h_i,v_j$ do not intersect for $i\neq j$.
By slightly tweaking their $x$-coordinates, any permutation of~$v_1,\dots,v_n$ can be realised as the lexicographic ordering~$<_\lex$.
Thus even for these very simple segment families, the ordered intersection graph $(\Fc,E,<_\lex)$ can be an arbitrary ordered matching, which are known to have unbounded twin-width (cf.~\cite{twin-width4}).

The solution is to have two layers of lexicographic ordering, at local and global scales.
That is, we first order the sets~$\Hc_{s,t}$ lexicographically by~$(s,t)$,
and then inside each~$\Hc_{s,t}$ order the segments~$\Hseg(x,y)$ lexicographically as in~$<_\lex$.
The same ordering applies to vertical segments.
Finally, all horizontal segments are placed before all vertical ones.
Crucially, in~$<_\seg$, and unlike~$<_\lex$, each set~$\Hc_{s,t}$ or~$\Vc_{s,t}$ is an interval.

For an APUS family~$\Fc$ with intersection graph~$G$, denote by $G^\seg \eqdef (\Fc,E,<_\seg)$ the same intersection graph ordered by~$<_\seg$.
Further, for $s,t \in \Z$ and $d \in \{\urc,\ulc,\lrc,\llc\}$,
let $G^\seg_{s,t,d} = (\Fc_{s,t,d},E_{s,t,d},<_\seg)$ be the restriction to~$\Fc_{s,t,d}$.
\begin{lemma}\label{lem:unitgrid-apus}
    If $\tww(G^\seg_{s,t,d}) \le k$ for all~$s,t,d$, then $\tww(G^\seg) \le f(k)$ for some function~$f$.
\end{lemma}
\begin{proof}
    Group the edge sets~$E_{s,t,d}$ into~4 classes according to the direction~$d$,
    i.e.\ define $\Ec_d \eqdef \bigcup_{s,t} E_{s,t,d}$.
    The sets~$\Ec_d$ for $d \in \{\urc,\ulc,\lrc,\llc\}$ are a partition of the edges,
    hence by \cref{thm:ordered-tww-ramsey}, it suffices to prove that each $(\Fc,\Ec_d,<_\seg)$ has bounded twin-width.
    
    We consider~$\Ec_\urc$, the others being similar.    
    In~$\Ec_\urc$, vertices of~$\Hc_{s,t}$ are only adjacent to those of~$\Vc_{s,t}$.
    Also, each~$\Hc_{s,t}$ and~$\Vc_{s,t}$ is an interval of~$<_\seg$.
    The intervals~$\Hc_{s,t}$ are ordered lexicographically by~$(s,t)$ between themselves,
    and the corresponding lexicographic ordering is used for the~$\Vc_{s,t}$.
    Finally, all~$\Hc_{s,t}$ come before all~$\Vc_{s',t'}$.
    This fits all the requirements of \cref{lem:ordered-components}, proving that~$(\Fc,\Ec_d,<_\seg)$ has twin-width at most~$2k+2$.
\end{proof}

\subsection{APUS graphs within a unit square}
Let us now focus on each~$G^\seg_{s,t,d}$, to find either bounded twin-width or a transversal pair, similar to \cref{sec:cag}.

Here, it must be mentioned that the graph~$G_{s,t,d}$ is the complement of a circular arc graph.
For~$G_{0,0,\urc}$, this is through the mapping $\Hseg(x,y) \mapsto [x,1+y]$ and $\Vseg(x,y) \mapsto [1+y,x]$ (the intervals being in $[0,2)$ modulo~2).
It is thus no surprise that the techniques of \cref{sec:cag} apply.
The results however cannot be applied as a blackbox: to construct transversal pairs when the matrix~$M_{\Fc_{s,t,d}}$ has a large grid, we may also need segments outside~$\Fc_{s,t,d}$.

When~$M_{\Fc_{s,t,d}}$ has no $k$-grid, its intersection graph has bounded twin-width.
By symmetry, it is enough to prove the result for~$\Fc_{0,0,\urc}$.
The only difference with \cref{lem:cag-gridthm} is that we now care about the ordered graph, with ordering~$<_\seg$.
\begin{lemma}\label{lem:gridthm-apus}
    Let~$\Fc$ be a family of APUS of the form $\Hseg(x,y)$ or $\Vseg(x,y)$ with $x,y \in [0,1)$.
    If the matrix~$M_\Fc$ has no $k$-grid, then $\tww(G_\Fc^\seg) \le f(k)$ for some function~$f$.
\end{lemma}
\begin{proof}
    Firstly, since the segments of~$\Fc$ have coordinates in the unit square~$[0,1)^2$,
    the ordering~$<_\seg$ by definition coincides with the lexicographic ordering~$<_\lex$ on~$\Fc$.

    Let~$X,Y \subset \R$ denote the columns and rows of~$M_\Fc$, i.e.\ the sets of $x$- and $y$-coordinates of endpoints of segments in~$\Fc$ respectively.
    We think of~$\Fc$ as a binary relation in $X \times Y$, where a segment~$h_{x,y}$ or~$v_{x,y}$ is interpreted as the pair~$(x,y)$.
    We order~$X \cup Y$ as $X < Y$, with the natural ordering inside each of them.
    Then the adjacency matrix $\adj(\Fc,<)$ consists of~$M_\Fc$, plus some empty zones to extend the columns to~$Y$ and rows to~$X$.
    Thus~$\adj(\Fc,<)$ has no $k$-grid, and by \cref{lem:incidence-matrix-grid}, the incidence matrix $\inc(\Fc,<,<_\lex)$ has no $4k$-grid.

    Construct now the following ordered graph $H = (V,E,<)$.
    The vertex set is $V \eqdef X \uplus Y \uplus \Fc$, ordered as $X < Y < \Fc$, with the natural ordering inside~$X,Y$, and the lexicographic ordering~$<_\lex$ inside~$\Fc$.
    The edges in~$E$ connect each segment~$h_{x,y}$ or~$v_{x,y}$ in~$\Fc$ to the coordinates~$x \in X$ and~$y \in Y$ of its endpoint.
    Thus the adjacency matrix of~$E$ consists of $\inc(\Fc,<,<_\lex)$ and its transpose, arranged as diagonal blocks.
    One may check that this matrix has no $(4k+1)$-grid,
    hence by \cref{thm:grid-tww}, $H$ has twin-width bounded by some~$f(k)$.

    To conclude, it suffice to show that there is a fixed first-order transduction~$\Phi$ such that $G_\Fc^\lex \in \Phi(H)$.
    Indeed, this implies by \cite[Theorem~8.1]{Bonnet2022twinwidth1} that $\tww(G_\Fc^\lex) \le g(\tww(H) \le g(f(k))$ for some function~$g$ depending only on~$\Phi$, proving the result.
    The transduction~$\Phi$ proceeds as follows given~$H$:
    \begin{itemize}
        \item Firstly, non-deterministic colouring is used to distinguish the subsets of vertices~$X$, $Y$, and~$\Fc$,
        and to indicate whether a vertex~$w \in \Fc$ is an horizontal or vertical segment.
        
        \item Next, observe that since the coordinates~$x,y,x',y'$ are restricted to~$[0,1)$, the segments~$\Hseg(x,y)$ and~$\Vseg(x',y')$ intersect if and only $x \le x'$ and $y' \le y$.
        To test if two segments~$s,s' \in \Fc$ intersect,
        it thus suffices to check that (1) one is horizontal and the other vertical (using the colours given at step~1),
        and (2) their coordinates satisfy the former inequalities.
        The later is checked by following edges of~$E$ from~$s,s'$ to their coordinates, and comparing them with the ordering~$<$.
        
        All of the above can be expressed as a first-order formula~$\phi(s,s')$.
        
        \item Finally, vertices of~$X,Y$ are deleted, and the edge set is redefined by the formula~$\phi(h,v)$ above, while keeping the same ordering~$<$ whose restriction to~$\Fc$ is~$<_\lex$.
    \end{itemize}
    Assuming that the correct non-deterministic choices are done in the colouring step,
    it is clear that~$\phi(s,s')$ holds if and only if~$s,s'$ intersect,
    and the structure obtained is thus $G_\Fc^\lex = G_\Fc^\seg$, proving the result.
\end{proof}

We now show that when~$M_{\Fc_{s,t,d}}$ does contain a grid, then~$G_\Fc$ contains a transversal pair.
As in \cref{sec:cag}, this requires some minimality hypothesis.

Here, we assume discrete coordinates as in \cref{sec:prelim:geom},
i.e.\ in an APUS family~$\Fc$, the coordinates of segments are non-negative multiples of some $\eta(\Fc)>0$.
Subject to this condition, we call~$\Fc$ \emph{minimized} if its segments cannot be pushed towards the bottom-left while preserving the intersection graph.
Formally, for a fixed choice of $\eta = \eta(\Fc)$, we say that~$\Fc$ is minimized if there is no different family~$\Fc'$ with the same $\eta = \eta(\Fc')$ and with a bijection~$\phi : \Fc \to \Fc'$
that preserves the vertical/horizontal orientation, is an isomorphism of the intersection graphs,
and such that the endpoint of~$\phi(s)$ is coordinate-wise no larger than that of~$s$.

\begin{lemma}\label{lem:minimal-apus}
    In a minimized APUS family~$\Fc$, for any horizontal $\Hseg(x,y) \in \Fc$, there are:
    \begin{enumerate}
        \item \label{item:horiz-min} a vertical segment~$\Vseg(x-\eta,y')$ or~$\Vseg(x+1,y')$ with $y' \le y \le y'+1$, unless~$x=0$, and
        \item \label{item:vert-min} a vertical segment~$\Vseg(x',y-1-\eta)$ or~$\Vseg(x',y)$ with $x \le x' \le x+1$, unless~$y=0$,
    \end{enumerate}
    and similarly when flipping the roles of vertical and horizontal, and of $x$- and $y$-coordinates.
\end{lemma}
\begin{proof}
    If that is not the case, $\Hseg(x,y)$ can be moved to the left (in case~1) or the bottom (in case~2) by~$\eta$ without changing the intersection graph, contradicting minimality.
\end{proof}

\begin{lemma}\label{lem:transversal-apus}
    Consider a minimized APUS family~$\Fc$, and its restriction~$\Fc_{s,t,d}$ for some~$s,t \in \Z$ and $d \in \{\urc,\ulc,\lrc,\llc\}$.
    There is a function $f: \N \rightarrow \N$ such that if~$M_{\Fc_{s,t,d}}$ has an $f(k)$-grid, then $G_\Fc$ contains a transversal pair~$T_k$.
\end{lemma}
\begin{proof}
    Take an $f(k)$-grid in~$M_{\Fc_{s,t,d}}$, for~$f(k)$ some large function of~$k$ to be determined later.
    It consists of points $(p_{i,j})_{i,j \in [f(k)]}$ as described in \cref{subsec:grids}.
    Each~$p_{i,j}$ corresponds in~$\Fc$ to either a horizontal segment with left endpoint~$p_{i,j}$, or a vertical segment with bottom endpoint at~$p_{i,j}$.
    By \cref{thm:grid-ramsey}, if this grid has size at least $R(g(k),2)$, then we can find a $g(k)$-subgrid consisting only of horizontal segments, or only of vertical segments.
    We thus fix $f(k) = R(g(k),2)$ (for~$g(k)$ to be determined later), and assume by symmetry that we find a $g(k)$-subgrid consisting only of horizontal segments $(h_{i,j})_{i,j \in [g(k)]}$ in~$\Fc_{s,t,d}$.
    Recall that the horizontal segments of~$\Fc_{s,t,d}$ are all in~$\Hc_{s,t}$ or all in~$\Hc_{s-1,t}$, depending on the direction~$d$.
    Assume without loss of generality this is~$\Hc_{s,t}$.
    Thus the segments~$h_{i,j}$ have their left endpoint in the square $[s,s+1) \times [t,t+1)$.

    We may also assume that none of the segments~$h_{i,j}$ have an endpoint at $x = 0$ or $y = 0$:
    This is ensured by simply removing the first row and first column of the grid.
    Then, since~$\Fc$ is minimized, \cref{lem:minimal-apus} gives that for each $h_{i,j} = \Hseg(x,y)$ in~$\Fc$, there is some segment $\hormin(h_{i,j}) \in \Fc$
    ensuring the horizontal minimality of~$h_{i,j}$ (i.e.\ blocking it from moving left), of the following form:
    \begin{itemize}
        \item either $\hormin(h_{i,j}) = \Vseg(x - \eta, y') \in \Fc$, with $y' \leq y \leq y' + 1$,
        \item or $\hormin(h_{i,j}) = \Vseg(x + 1, y') \in \Fc$ with $y' \leq y \leq y' + 1$.
    \end{itemize}
    Similarly, there is some $\vermin(h_{i,j}) \in \Fc$ ensuring vertical minimality of~$h_{i,j}$ (blocking it from moving down), of the following form:
    \begin{itemize}
        \item either $\vermin(h_{i,j}) = \Vseg(x', y) \in \Fc$ with $x \leq x' \leq x + 1$,
        \item or $\vermin(h_{i,j}) = \Vseg(x', y-1-\eta) \in \Fc$ with $x \leq x' \leq x + 1$.
    \end{itemize}
    Finally, we care about whether $\hormin(h_{i,j})$ intersects the horizontal line with coordinate $y=t$, or $y=t+1$.
    Note that this is not relevant for~$\vermin(h_{i,j})$: segments $\Vseg(x',y)$ will always intersect the line $y=t+1$,
    while $\Vseg(x',y-1-\eta)$ intersects the line $y=t$.

    Colour the points~$h_{i,j}$ of the $g(k)$-grid according to all the previous cases, i.e.\ to indicate:
    \begin{enumerate}
        \item whether horizontal minimality of~$h_{i,j}$ is ensured on the left or the right,
        \item whether vertical minimality of~$h_{i,j}$ is ensured above or below, and
        \item whether or not $\hormin(h_{i,j})$ intersects the line $x=s+1$.
    \end{enumerate}

    This gives~8 possible colours in total.
    For an appropriate choice of~$g(k)$, \cref{thm:grid-ramsey} yields a monochromatic $2k$-subgrid.
    Let us assume that as outcome of \cref{thm:grid-ramsey}, this $2k$-grid satisfies the following conditions:
    \begin{enumerate}
        \item For each $h_{i,j} = \Hseg(x,y)$, horizontal minimality is ensured on the left, meaning that $\hormin(h_{i,j}) = \Vseg(x-\eta,y')$ with $y' \le y \le y'+1$.
        \item Similarly, vertical minimality of $h_{i,j} = \Hseg(x,y)$ is ensured on the bottom, meaning that $\vermin(h_{i,j}) = \Vseg(x',y-1-\eta)$ with $x \le x' \le x+1$.
        \item Finally, the segments $\hormin(h_{i,j})$ all intersect the line $x=s+1$.
    \end{enumerate}
    This case is depicted in \cref{fig:construction_apus}.
    All other cases can be handled similarly, reversing the indices in the rest of the construction appropriately.

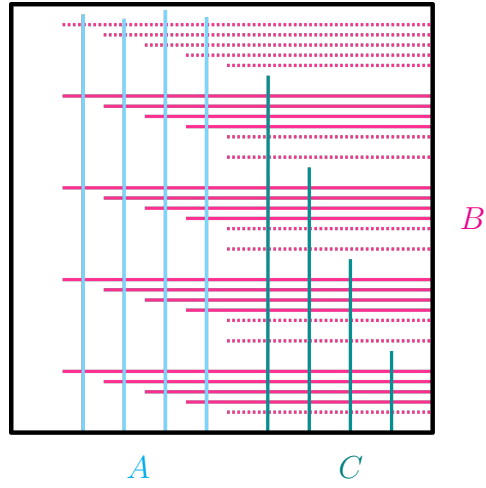
\begin{figure}
    \centering
    \begin{tikzpicture}
        \node at (0,0) {\includegraphics[width=0.4\linewidth]{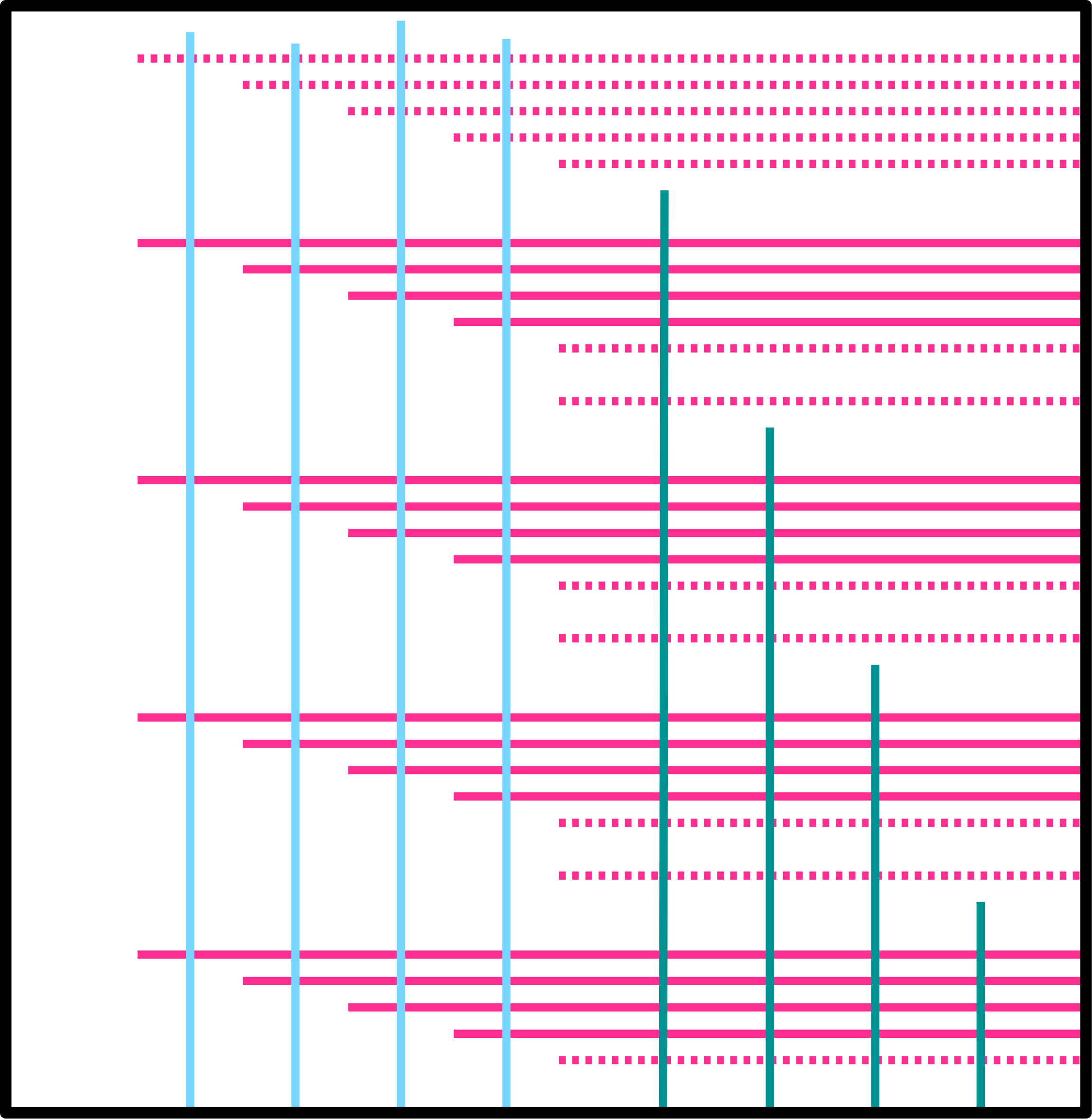}};
        \node at (-1.1,-3.3) {\color{cyan} \Large $A$};
        \node at (1.7,-3.3) {\color{teal} \Large $C$};
        \node at (3.3,0) {\color{magenta} \Large $B$};
        
    \end{tikzpicture}
    \caption{Transversal pair $T_4$ as APUS graph. 
    The {\color{magenta}magenta} lines are part of an $8$-grid.
    Some of them (solid lines) are picked as~$B$. The others (dotted) are used to define~$A,C$:
    the {\color{cyan}cyan} (resp.\ {\color{teal}teal}) lines forming~$A$ (resp.~$C$) ensure horizontal (resp.\ vertical) minimality of said dotted lines.
    }
    \label{fig:construction_apus}
\end{figure}

    Recall that by definition of a grid, the $x$- and $y$-coordinates of the~$h_{i,j}$ satisfy
    \[ \text{if $i<i'$, then $\Xco(p_{i,j}) < \Xco(p_{i',j'})$} \qquad \text{and} \qquad \text{if $j<j'$, then $\Yco(p_{i,j}) < \Yco(p_{i',j'})$}. \]
    From this, and the assumptions on~$\hormin(h_{i,j})$, the following is immediate for $i,i',j,j' \in [2k]$:
    \begin{align}
        & \text{if $i<i'$, then $\hormin(h_{i,j})$ does not intersect $h_{i',j'}$, and} \\
        & \text{if $i>i'$ and $j>j'$, then $\hormin(h_{i,j})$ does intersect $h_{i',j'}$.}
    \end{align}
    Thus, if we take $a_i \eqdef \hormin(h_{2i,2k})$ and $b_{i',j'} \eqdef h_{2i'-1,2j'-1}$,
    then we have that~$a_i$, $b_{i',j'}$ intersect if and only if~$i \ge i'$.

    Similarly, from the assumptions on~$\vermin(h_{i,j})$, we have for $i,i',j,j' \in [2k]$:
    \begin{align}
        & \text{if $j<j'$, then $\vermin(h_{i,j})$ does not intersect $h_{i',j'}$, and} \\
        & \text{if $i>i'$ and $j>j'$, then $\vermin(h_{i,j})$ does intersect $h_{i',j'}$.}
    \end{align}
    Thus, for $c_j \eqdef \vermin(h_{2k,2j})$ (and the same $b_{i',j'} \eqdef h_{2i'-1,2j'-1}$),
    we have that~$c_j$, $b_{i',j'}$ intersect if and only if $j \ge j'$.
    
    Up to reversal of the indices of the~$a_i$s, the vertices~$a_i$, $b_{i,j}$, and~$c_j$ form a transversal pair~$T_k$ as described in \cref{def:transversal-pair}.
\end{proof}

Our main theorem for APUS graphs follows from \cref{lem:unitgrid-apus,lem:gridthm-apus,lem:transversal-apus}, along the same lines as \cref{thm:delineation_circular_arc},
with the additional step of splitting into unit squares (\cref{lem:unitgrid-apus}).
Since computing APUS representation of graphs is NP-hard~\cite{mustata2013unitgrid},
some subtleties arise regarding the complexity of FO model checking:
we assume graphs to be given by their APUS representation,
and use that it is easy to compute an APUS representation of transversal pairs as in \cref{fig:construction_apus}.%
\begin{theorem}\label{thm:delineation-apus}
    Let~$\Cc$ be a subclass of non-degenerate APUS graphs.
    Then~$\Cc$ has bounded twin-width if and only if it is monadically dependent.
    Further, when~$\Cc$ is hereditary, and assuming $\fpt \neq \aw$,
    these are also equivalent to FO model checking in~$\Cc$ being FPT, when graphs in~$\Cc$ are given by some APUS representation.
\end{theorem}
\begin{proof}
    Consider all discrete and minimized APUS representations~$\Fc$ of graphs in~$\Cc$.
    For each of them, consider further the subfamilies~$\Fc_{s,t,d}$ for $s,t, \in \Z$ and $d \in \{\urc,\ulc,\lrc,\llc\}$,
    and their endpoint matrices~$M_{\Fc_{s,t,d}}$.
    Call~$\Mc$ the class of matrices~$M_{\Fc_{s,t,d}}$ over all choices of~$\Fc,s,t,d$ as above.
    If there is a~$k$ such that no matrix in~$\Mc$ has a $k$-grid,
    then the ordered graphs~$G^\seg_{\Fc_{s,t,d}}$ have twin-width bounded by some~$f(k)$ by \cref{lem:gridthm-apus}.
    It follows by \cref{lem:unitgrid-apus} that the full intersection graphs~$G^\seg_\Fc$ also have bounded twin-width, i.e.~$\Cc$ has bounded twin-width.
    On the other hand, if~$\Mc$ has matrices with arbitrarily large grids,
    then \cref{lem:transversal-apus} gives arbitrarily large transversal pairs in graphs of~$\Cc$, hence~$\Cc$ is monadically independent.

    Assuming that~$\Mc$ has no $k$-grid, observe that given any APUS representation of a $G \in \Cc$,
    one can in polynomial time find a discrete and minimal representation~$\Fc$, and compute the ordered graph~$G^\seg_\Fc$ with bounded twin-width.
    Since this is an ordered graph, \cite[Theorem~7]{twin-width4} gives an FPT algorithm to find a contraction sequence of bounded width.
    It follows by \cref{thm:model-checking} that FO model checking in~$\Cc$ is FPT.

    On the other hand, if~$\Cc$ is hereditary and contains arbitrarily large transversal pairs (and thus is monadically independent),
    then \cref{thm:independent-hardness} shows that FO model checking in~$\Cc$ is \aw-hard, \emph{when the input is given as a graph $G \in \Cc$ itself.}
    We want to show that the problem remains hard if the graph is given by some APUS representation.

    Note that~$\Cc$ is a class of bipartite graphs.
    In general, transversal pairs are defined up to some unconstrained vertex pairs, i.e.\ pairs which may or not be an edge (see \cref{def:transversal-pair}).
    For bipartite graphs however, one may check that these unconstrained pairs must all be non-edges.
    Thus it is meaningful to talk about \emph{the} transversal pair~$T_k$ of order~$k$, uniquely defined in~$\Cc$.
    For any~$T_k$, or any induced subgraph thereof, \cref{fig:construction_apus} gives an APUS representation which can be computed in polynomial time (this construction was already known in \cite[Figure~3.4]{twin-width2}).
    This gives an FPT reduction from model checking on transversal pairs given as graphs (which is \aw-hard), to model checking on graphs in~$\Cc$ given by some APUS representation.
\end{proof}

\subsection{Segments with a fixed set of permissible lengths}
We now consider the generalisation of APUS from the last item of \cref{thm:main}:
axis-parallel segments whose lengths belong to a fixed, finite set $X \subseteq \R_+$.
By compactness, these lengths can be assumed to be rational, and thus up to rescaling,
this is equivalent to requiring segments to have integer lengths in $\{1,\dots,L\}$ for some fixed $L \in \N$.

\begin{theorem}\label{thm:seg-fixed-length-set}
    Let~$\Cc$ be a class of non-degenerate, axis-parallel segments graphs, with segment lengths in~$[L]$ for some fixed~$L$.
    Then~$\Cc$ has bounded twin-width if and only if it is monadically dependent.
    Further, when~$\Cc$ is hereditary, and assuming $\fpt \neq \aw$, these are also equivalent to FO model checking in~$\Cc$ (given by segment representations) being FPT.
\end{theorem}
\begin{proof}[Sketch of proof.]
    Let~$\Fc$ be a family of axis-parallel segments with lengths in~$[L]$.
    We once again partition the segments according to their direction, the unit square containing their first endpoint, and now also their length.
    That is, define~$\Hc_{s,t,\ell}$ to contain all horizontal segments in~$\Fc$ of length~$\ell$ with their left endpoint in $[s,s+1) \times [t,t+1)$, and similarly~$\Vc_{s,t,\ell}$ for vertical ones.
    Clearly segments of~$\Hc_{s,t,\ell}$ and~$\Vc_{s',t',\ell'}$ can only intersect when $|s-s'|$ and $|t-t'|$ are at most~L.
    We also generalise the global ordering~$<_\seg$ of segments to account for segment lengths.
    The ordering~$<_\seg$ sorts segments according to, in that order: (1)~direction, (2)~segment length, (3)~the square containing the endpoint, lexicographically, and finally (4)~the endpoint coordinate, lexicographically.

    Generalising \cref{lem:unitgrid-apus}, if the ordered intersection graph induced by~$\Hc_{s,t,\ell} \cup \Vc_{s',t',\ell'}$ has twin-width at most~$k$ for all parameters $s,t,s',t',\ell,\ell'$,
    then the full ordered intersection graph $(\Fc,E,<_\seg)$ has twin-width bounded by a function of~$k$ and~$L$.
    The proof is essentially the same:
    Each edge of~$G_\Fc$ joins some sets~$\Hc_{s,t,\ell}$ and~$\Vc_{s',t',\ell'}$,
    with $i \eqdef s-s'$ and $j \eqdef t-t'$ both in~$\{-L,\dots,L\}$.
    Then, we classify all edges of~$G_\Fc$ according to the offsets~$i,j$ and the lengths~$\ell,\ell'$ as above.
    This is a partition of the edge set into~$O(L^4)$ groups, and by \cref{thm:ordered-tww-ramsey}, it suffices to prove that each of these classes of edges gives an ordered subgraph of bounded twin-width.
    The latter is obtained by using \cref{lem:ordered-components} exactly as in the proof of \cref{lem:unitgrid-apus}.

    Next, with essentially no modifications, \cref{lem:gridthm-apus,lem:transversal-apus} can be adapted to show that:
    \begin{itemize}
        \item If the endpoints of segments in $\Hc_{s,t,\ell} \cup \Vc_{s',t',\ell}$ do not form a grid,
        then the corresponding induced subgraph, ordered by~$<_\seg$, has bounded twin-width, and
        \item if however the endpoints do form an $f(k)$-grid (for appropriate~$f$), and~$\Fc$ has been minimized, then~$G_\Fc$ contains the transversal pair~$T_k$.
    \end{itemize}
    These results combine to prove the theorem in the same way as \cref{thm:delineation-apus}.
\end{proof}

\section{Obstructions to delineation}\label{sec:obstructions}
In this section, we construct graphs with unbounded twin-width, but for which FO model checking is FPT.
We use them to show that delineation fails for 1.5D-terrain visibility graphs,
and also when relaxing any of the hypotheses of \cref{thm:delineation-apus}.

Precisely, we construct a class~$\Hc$ of graphs combining half-graphs and paths of slightly more than constant length,
and prove that is has unbounded twin-width, but bounded \emph{merge-width}~\cite{merge-width}.
The latter implies FPT model checking and monadic dependence.
The class~$\Hc$ differs significantly from the class of all subcubic graphs used as the basic example of non-delineation in~\cite{bonnet2022delineation}:
any monadically stable (i.e.\ that does not transduce all half-graphs) subclass of~$\Hc$ has bounded twin-width.
This gives a concrete counter-example to \cite[Conjecture~6.6]{bonnet2022delineation}.
Samuel Braunfeld previously found a different, non-explicit counter-example to this conjecture (personal communication).

\subsection{Construction}
Let~$\Perm_n$ denote the set of permutations of~$[n]$.
Given $\sigma \in \Perm_n$ a length $\ell > 0$, define~$H_\sigma^\ell$ as follows (see \cref{fig:halfgraph-longpath}).
Create vertices $a_i,b_i,c_i,d_i$ for~$i \in [n]$, forming two disjoint half-graphs with edges~$a_ib_j$ and~$c_id_j$ for all~$i \le j$.
Then, connect~$b_i$ to~$c_{\sigma(i)}$ by a path of length~$\ell$.
For a function $f : \N \to \N$, we define the class of graphs $\Hc_f \eqdef \{H_\sigma^{f(n)} : n \in \N,\ \sigma \in \Perm_n\}$.

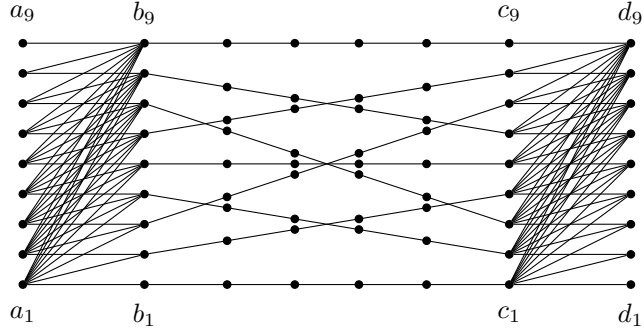
\begin{figure}[htb]
    \centering
    \begin{tikzpicture}[scale=0.8]
        \tikzstyle{vertex}=[black,fill,draw,circle,inner sep=0pt, minimum width=1mm]
        \foreach \i in {1,...,9}{
            \node[vertex] (a\i) at (0,0.5*\i) {};
            \node[vertex] (b\i) at (2,0.5*\i) {};
            \node[vertex] (c\i) at (8,0.5*\i) {};
            \node[vertex] (d\i) at (10,0.5*\i) {};
        }
        \node at (0,0) {$a_1$};
        \node at (2,0) {$b_1$};
        \node at (8,0) {$c_1$};
        \node at (10,0) {$d_1$};

        \node at (0,5) {$a_9$};
        \node at (2,5) {$b_9$};
        \node at (8,5) {$c_9$};
        \node at (10,5) {$d_9$};
        
        \foreach \i in {1,...,9}{
            \foreach \j in {\i,...,9}{
                \draw (a\i) -- (b\j);
                \draw (c\i) -- (d\j);
            }
        }

        \foreach \i/\j in {1/1,2/4,3/7,4/2,5/5,6/8,7/3,8/6,9/9}{
            \draw (b\i) -- (c\j)
                node[vertex, pos=0.22] {}
                node[vertex, pos=0.41] {}
                node[vertex, pos=0.59] {}
                node[vertex, pos=0.78] {};
        }
    \end{tikzpicture}
    \caption{The graph $H_\sigma^5$ for the permutation $\sigma = 147258369$.}
    \label{fig:halfgraph-longpath}
\end{figure}

When~$f$ grows slowly, we show that~$\Hc_f$ has unbounded twin-width by transducing $f(n)$-subdivisions of bicliques of size~$\sqrt{n}$, known to have unbounded twin-width~\cite{twin-width2}.%
\begin{fact}\label{fact:logn-subdiv-tww}
    For $f(n) = o(\log n)$ a positive function, $\Hc_f$ has unbounded twin-width.
\end{fact}
\begin{proof}
    Denote by~$K_{n,n}$ the biclique with~$n$ vertices on each sides,
    and by~$K_{n,n}^{(\ell)}$ the graph obtained by replacing each edge of this biclique by a path of length~$\ell$.
    For a positive function~$g$, consider the class $\Kc_g \eqdef \{K_{n,n}^{(g(n))} : n \in \N\}$.
    A simple variant of \cite[Theorem~6.2]{twin-width2} shows that~$\Kc_g$ has unbounded twin-width whenever $g(n) = o(\log n)$ and $g(n) \ge 2$.
    We will show that~$\Hc_f$ transduces~$\Kc_g$ for $g(n) \eqdef f(n^2)+2$.
    Note that $f(n) = o(\log n)$ implies $f(n^2) = o(\log n)$, so~$\Kc_g$ indeed has unbounded twin-width.
    By \cite[Theorem~8.1]{Bonnet2022twinwidth1}, this transduction implies that~$\Hc_f$ also has unbounded twin-width.

    Specifically, we construct a first-order transduction~$\Phi$ satisfying the following:
    If~$\sigma$ is the universal permutation on~$n^2$ elements defined as $\sigma(in+j+1) = jn+i+1$ for $i,j \in \{0,\dots,n-1\}$,
    then $\Phi(H_\sigma^\ell) \ni K_{n,n}^{(\ell+2)}$ for any~$\ell \in \N$.
    In what follows, $N \eqdef n^2$ denotes the size of~$\sigma$.
    
    The transduction~$\Phi$ proceeds as follows. Recall here that the composition of several transductions is also a transduction (\cref{lem:transduction-comp}).
    \begin{itemize}
        \item First, pick every $n$th vertex in~$A$ and~$D$,
        that is define $x_i \eqdef a_{(i-1)n+1}$ and $X \eqdef \{x_i : i \in [n]\}$,
        and symmetrically $y_i \eqdef d_{in}$ and $Y \eqdef \{y_i : i \in [n]\}$.
        These will become the high degree vertices in~$K_{n,n}^{(\ell+2)}$.
        Vertices of $A \setminus X$ and $D \setminus Y$ will not be used, the transduction~$\Phi$ starts by deleting them
        (taking any induced subgraph is a transduction: one marks the vertices to be kept with non-deterministic colouring, and only keep these in the interpretation step).
        \item Next, use non-deterministic colouring to mark each of the subsets~$X,Y$, as well as $B \eqdef \{b_i\}_{i \in [N]}$, and $C \eqdef \{c_i\}_{i \in [N]}$.
        \item Our goal is now to keep the edge $x_i b_j$ if and only if $(i-1)n < j \le in$.
        Observe that this inequality holds if and only if (1) $x_i b_j$ is an edge, and (2) $x_i$ is minimal with this property,
        in the sense any other $x' \in X$ adjacent to~$b_j$ satisfies $N(x_i) \subseteq N(x')$.
        These two conditions can be expressed by an FO formula~$\phi(x_i,b_j)$, allowing to keep only the desired edges with an interpretation (see \cref{fig:halfgraph-longpath-transduction}).

        A similar argument can be used to keep the edge~$c_jy_i$ if and only if $(i-1)n < j \le in$.
        The edges of the paths between~$B$ and~$C$ are not modified in this step.
    \end{itemize}
    
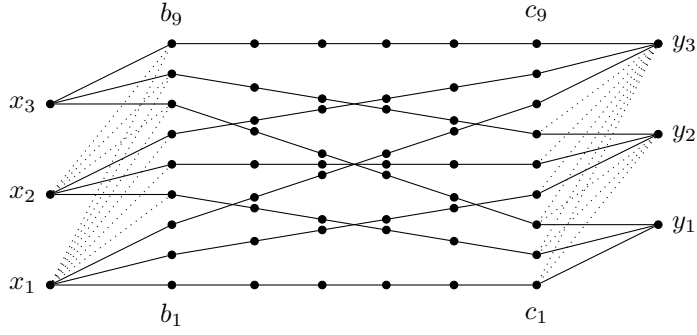
\begin{figure}[tb]
    \centering
    \begin{tikzpicture}[scale=0.8]
        \tikzstyle{vertex}=[black,fill,draw,circle,inner sep=0pt, minimum width=1mm]
        \foreach \j in {1,...,9}{
            \node[vertex] (b\j) at (2,0.5*\j) {};
            \node[vertex] (c\j) at (8,0.5*\j) {};
        }
        \foreach \i in {1,...,3}{
            \pgfmathtruncatemacro{\a}{\i*3-2}
            \pgfmathtruncatemacro{\d}{\i*3}
            \node[vertex, label={left:$x_\i$}] (x\i) at (0,0.5*\a) {};
            \node[vertex, label={right:$y_\i$}] (y\i) at (10,0.5*\d) {};

            \foreach \j in {\a,...,\d}{
                \draw (x\i) -- (b\j);
                \draw (c\j) -- (y\i);
            }
            \foreach \j in {1,...,9}{
                \ifnum \j>\d
                    \draw[dotted] (x\i) -- (b\j);
                \fi
                \ifnum \j<\a
                    \draw[dotted] (y\i) -- (c\j);
                \fi
            }
        }

        \node at (2,0) {$b_1$};
        \node at (8,0) {$c_1$};
        \node at (2,5) {$b_9$};
        \node at (8,5) {$c_9$};

        \foreach \i/\j in {1/1,2/4,3/7,4/2,5/5,6/8,7/3,8/6,9/9}{
            \draw (b\i) -- (c\j)
                node[vertex, pos=0.22] {}
                node[vertex, pos=0.41] {}
                node[vertex, pos=0.59] {}
                node[vertex, pos=0.78] {};
        }
    \end{tikzpicture}
    \caption{Transduction from~$H_\sigma^5$ to the subdivided biclique $K_{3,3}^{(7)}$, for $\sigma = 147258369$.
        The unnecessary vertices $A \setminus X$ and $D \setminus Y$ have already been removed.
        The solid edges form~$K_{3,3}^{(7)}$, and the transduction only needs to remove the dotted edges of the half-graphs to obtain it.
    }
    \label{fig:halfgraph-longpath-transduction}
\end{figure}

    Assuming the `correct' non-deterministic choices happen,
    the resulting graph consists of two families of~$n$ stars with~$n$ branches, with~$x_i$s (resp.~$y_i$s) as centres and $b_j$s (resp.~$c_j$s) as branches.
    The choice of universal permutation~$\sigma$ ensures that each $x_i$-star is connected to each $y_{i'}$-star by a path of length~$\ell$, yielding the graph~$K_{n,n}^{(\ell+2)}$.
\end{proof}

On the other hand, when~$f$ tends to infinity, the class~$\Hc_f$ has bounded merge-width.
More generally, we prove the following.
\begin{lemma}\label{lem:mw-long-paths}
    Let~$G$ be a graph with $\mw_r(G) = k$, and $\ell > 2r$. Pick any number of pairs of vertices $(x_1,y_1),\dots,(x_m,y_m)$,
    and for each of them add a new path of length~$\ell$ from~$x_i$ to~$y_i$, with fresh internal vertices.
    Call~$G'$ the resulting graph. Then $\mw_r(G') \le \max(k+1,4)$.
\end{lemma}
\begin{proof}
    Call~$P_i$ the path of length $\ell$ connecting~$x_i$ and~$y_i$ for $i \in [m]$,
    and denote its vertices by $x_i = v_i^0 v_i^1 v_i^2 \dots v_i^{\ell-1} v_i^{\ell} = y_i$.
    In a first phase, we inductively resolve the paths~$P_i$, and merge their internal vertices into a single part.

    First positively resolve the edges in the path $x_1 v_1^1 v_1^2 v_1^3$, before merging~$\{v_1^1\},\{v_1^2\}$ into the same part.
    Next resolve $v_1^3 v_1^4$, and merge~$\{v_1^3\}$ with $\{v_1^1,v_1^2\}$,
    and repeat until all of~$P_1$ is resolved, and all its internal vertices are in the same part.
    We then proceed with~$P_2$ similarly, not forgetting to merge its internal vertices with those of~$P_1$.
    In general, before processing vertex~$v_i^j$ in~$P_i$,
    the edges that are resolved are all edges of $P_1,\dots,P_{i-1}$, and those of the prefix $x_1 \dots v_i^j$ of~$P_i$;
    the current partition has one part~$S$ containing all internal vertices of $P_1,\dots,P_{i-1}$ plus $v_i^1,\dots,v_i^{j-1}$,
    and all other vertices in singletons.
    The next steps are to positively resolve $v_i^j$ with $v_i^{j+1}$, and merge~$v_i^j$ into~$S$.
    Note that during this process, the radius-$r$ width is at most~4.
    When finished, all the internal vertices of all $P_i$s are merged into a single part~$S$, while all edges of~$P_i$ are resolved.
    
    We then proceed with the construction sequence for $G$ witnessing $\mw_r(G) = k$.
    Since the paths~$P_i$ have length more than~$r$, they cannot be used by a vertex $v \in V(G)$ to reach any new part besides~$S$.
    Also, for a vertex $v_i^j \in P_i$, all parts $r$-reachable from~$v_i^j$ must also be reachable from either~$x_i$ (if~$j \le r$) or from~$y_i$ (when $j \ge \ell-r > r$).
    Thus, the radius-$r$ width of this second part of the construction sequence is at most~$k+1$.
    
    At this point, all pairs in~$V(G)$ and all edges of $P_i$'s are resolved, and there are just two parts remaining:~$S$ and~$V(G)$.
    We merge them, and negatively resolve all remaining pairs to complete the construction sequence.
    This construction sequence for~$G'$ has width at most $\max(k+1,4)$.
\end{proof}

\begin{fact}\label{fact:subdiv-mw}
    For any function~$f$ which tends to infinity, the class $\Hc_f$ has bounded merge-width.
\end{fact}
\begin{proof}
    It is simple to check that half-graphs have merge-width at most~3 at any radius
    (one could also get a bound through clique-width using \cite[Theorem~7.1]{merge-width}).
    Since $H_\sigma^\ell$ is obtained from half-graphs by the process described in \cref{lem:mw-long-paths},
    it follows that $\mw_r(H_\sigma^\ell) \le 4$ whenever $\ell > 2r$.

    Now denote by~$g$ some inverse of~$f$, in the sense that $n \le g(f(n))$ for all~$n$.
    Now consider $H_\sigma^{f(n)} \in \Hc_f$ for some permutation $\sigma \in \Perm_n$.
    Remark that $|V(H_\sigma^{f(n)})| = n \cdot (f(n)+3)$, which is a trivial upper bound on the merge-width of~$H_\sigma^{f(n)}$ at any radius~$r$.
    If $f(n) \le 2r$, then by choice of~$g$ we have $n \le g(2r)$, which implies
    \[ \mw_r(H_\sigma^{f(n)}) \le n \cdot (f(n)+3) \le g(2r) \cdot (2r+3), \]
    while if~$f(n) > 2r$, then $\mw_r(H_\sigma^{f(n)}) \le 4$ as argued in the first paragraph.
    This proves that~$\Hc_f$ has bounded merge-width.
\end{proof}

Now pick for instance $f(n) \eqdef \log\log n$, and consider $\Hc \eqdef \Hc_f$.
\begin{corollary}
    \label{cor:not-delin1}
    The class~$\Hc$ has unbounded twin-width, but is monadically dependent, and FO model checking in~$\Hc$ is FPT.
\end{corollary}
\begin{proof}
    The class~$\Hc$ has unbounded twin-width by \cref{fact:logn-subdiv-tww}, and bounded merge-width by \cref{fact:subdiv-mw}.
    The latter implies monadically dependent and FPT model checking by \cref{thm:mw-model-checking}.
\end{proof}

\subsection{Non-delineation in segment graphs}
We can now show that all hypotheses of \cref{thm:seg-fixed-length-set} are required,
by showing that relaxing any of them allows to obtain the non-delineated class~$\Hc_f$ (for any function~$f$ with~$f(n) \ge 3$) as intersection graph.
\nondelinsegments*
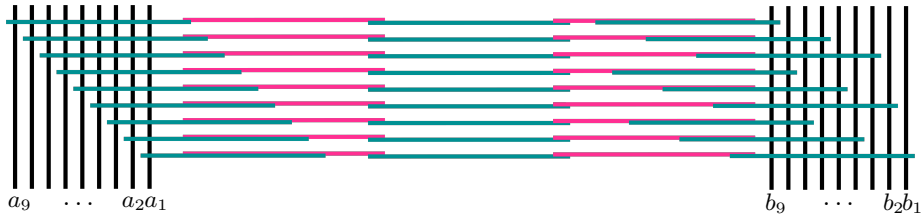
\begin{figure}[tbh]
    \centering
    \begin{tikzpicture}
        \node at (0,0) {\includegraphics[width=12cm]{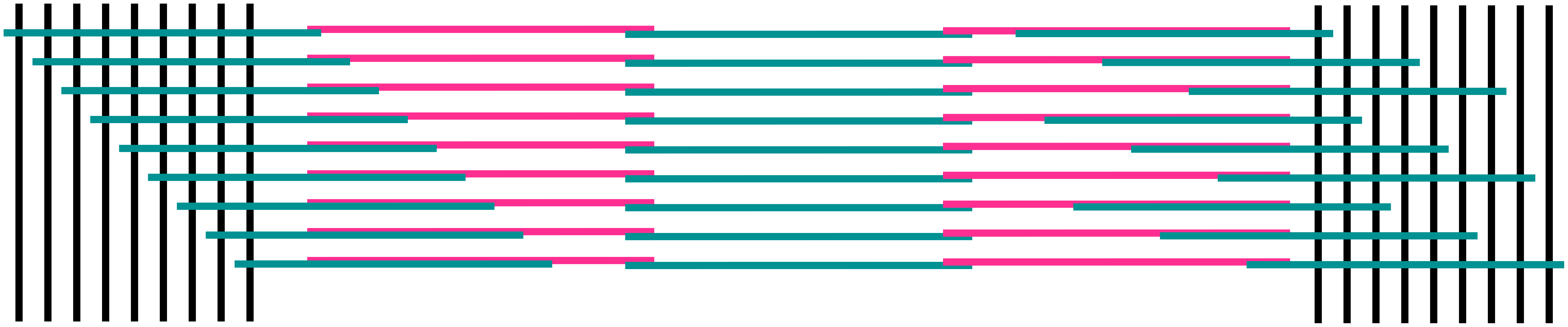}};
        \node at (-4.0,-1.4) {\small $a_1$};
        \node at (-4.3,-1.4) {\small $a_2$};
        \node at (-5,-1.4) {$\dots$};
        \node at (-5.8,-1.4) {\small $a_9$};

        \node at (5.95,-1.4) {\small $b_1$};
        \node at (5.7,-1.4) {\small $b_2$};
        \node at (5,-1.4) {$\dots$};
        \node at (4.15,-1.4) {\small $b_9$};
    \end{tikzpicture}
    \caption{A degenerate APUS whose intersection graph is $H_\sigma^5$ for $\sigma = 147258369$.}
    \label{fig:non_delineated_APUS}
\end{figure}
\begin{proof}
    \Cref{fig:non_delineated_APUS} shows how to construct any graph $H_\sigma^\ell$ (with $\ell \ge 2$) as an APUS graph if one allows intersections between horizontal segments.
    Thus the class~$\Hc_f$ of \cref{cor:not-delin1} is a class of degenerate APUS graphs, proving point~\ref{item:degen}.

    Almost the same construction can be use for points~\labelcref{item:twolengths,item:epsilon-directions}.
    For~\ref{item:twolengths}, we simulate the intersection between two horizontal segments~$h_1,h_2$ by adding a tiny vertical segment~$v$ intersecting only these two.
    This only doubles the length of the paths between the two half-graphs in the intersection graphs, which does not change the fact that the class is non-delineated.
    For~\ref{item:epsilon-directions}, the horizontal segments in \cref{fig:non_delineated_APUS} can be made slightly more and slightly less than horizontal alternatively,
    so that the family of segments is now in general position, while keeping the same intersection graph.

    Finally, to prove point~\labelcref{item:epsilon-lengths}, we once again obtain any graph~$H_\sigma^\ell$, but with a slightly different construction, illustrated in \cref{fig:var-length-not-delineated}.
    \begin{figure}[tbh]
    \centering
    \includegraphics[width=0.5\linewidth]{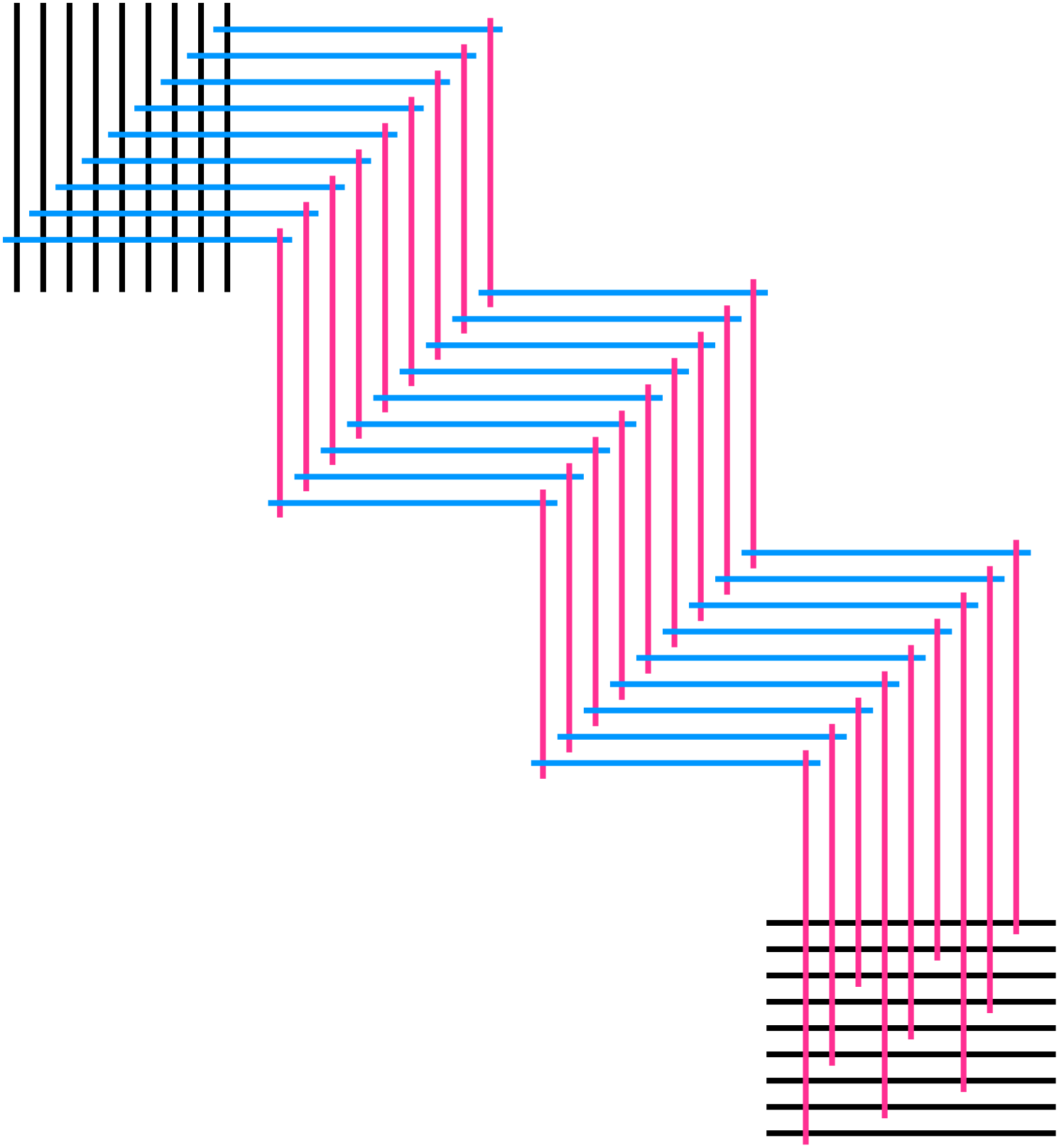}
    \caption{%
        Construction of a graph~$H^5_\sigma$ as intersection graph of axis-parallel segments whose lengths are between~1 and~2.
        The only segments of non-unit length are in the bottom-most vertical group, to represent an arbitrary permutation as half-graph.
        By packing each group of segments tighter together, the segment lengths needed in this construction can be made arbitrarily close to~1.
    }
    \label{fig:var-length-not-delineated}
\end{figure}
\end{proof}

\subsection{1.5D-Terrain visibility graphs}
A variant of this construction can be used for terrain visibility graphs:
\terrainnotdelin*

\begin{figure}[tbh]
    \centering
    \begin{tikzpicture}    
        \tikzstyle{vertex}=[black,fill,draw,circle,inner sep=0pt, minimum width=0.5mm]
        \def\cx#1{0.005*(#1)*(#1) - 0.1*(#1)}
        \def\cy#1{0.5*(#1) + 1}
        \def\bx#1{(#1)}
        \def\by#1{0.2*(#1) - 0.01*(#1)*(#1)}
        
        \foreach \i in {1,...,9}{
            \node[vertex] (c\i) at ({\cx\i},{\cy\i}) {};
        }
        \foreach \i/\j in {1/3,2/6,3/9,4/2,5/5,6/8,7/1,8/4,9/7}{
            \pgfmathsetmacro{\x}{\bx{\i-0.1}}
            \pgfmathsetmacro{\y}{\by{\i-0.1}}
            \pgfmathsetmacro{\xx}{\bx{\i+0.1}}
            \pgfmathsetmacro{\yy}{\by{\i+0.1}}
            \pgfmathsetmacro{\tx}{\cx\j-0.1}
            \pgfmathsetmacro{\ty}{\cy\j+0.3}
        
            \pgfmathsetmacro{\slope}{(\ty - \y) / (\tx - \x)}

            \draw (\x,\y) node[vertex] (a\i1){} -- ++(0.1,0.1*\slope) node[vertex] (b\i){} -- (\xx,\yy) node[vertex] (a\i2){};
            \draw[dotted] (\x,\y) -- (\tx,\ty);
        }
        
        \foreach \i in {1,...,8}{
            \pgfmathtruncatemacro{\j}{\i+1}
            \draw (a\i2) -- (a\j1)
                node[vertex, pos=0.33] (a\i3) {}
                node[vertex, pos=0.67] (a\i4) {};
            \draw (c\i) -- (c\j);
        }

        \draw (a11) -- (c1);

        \node[vertex, label=left:$c_1$] at (c1) {};
        \node[vertex, label=left:$c_9$] at (c9) {};
        \node[vertex, pin=below left:$a_0$] at (a11) {};
        \node[vertex, pin=below:$b_1$] at (b1) {};
        \node[vertex, pin=below right:$a_1$] at (a12) {};
        \node[vertex, pin=below:$b_9$] at (b9) {};
        \node[vertex, pin=below left:$a_{8\ell}$] at (a91) {};
        \node[vertex, pin=below right:$a_{8\ell+1}$] at (a92) {};
    \end{tikzpicture}
    \caption{Construction of non-delineated terrain visibility graphs.
        Dotted lines represent the `horizon' for each~$b_i$, which can be controlled by moving~$b_i$ up or down.
        This allows the half-graph between the~$b_i$s and $c_j$s to represent any permutation of the~$b_i$s, here $369258147$.
    }
    \label{fig:non-delineated-terrain}
\end{figure}
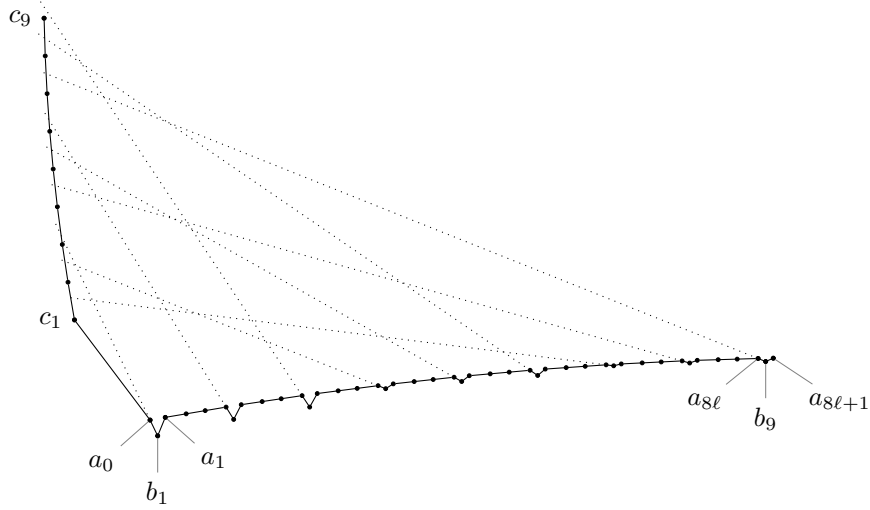
To prove this result, we construct a variant of the graphs~$H_\sigma^\ell$ as a terrain visibility graph.
Consider the following terrain, depicted in \cref{fig:non-delineated-terrain}.
On the left side, construct an almost vertical path $C = \{c_1,\dots,c_n\}$, with~$c_n$ top-left most and~$c_1$ bottom-right most.
This path is slightly convex, so that~$C$ is a clique in the visibility graph.
To the right of and below~$C$, construct an almost horizontal path $A = \{a_1,\dots,a_{\ell (n-1)+1}\}$,
slightly concave so that $a_1,\dots,a_{\ell (n-1)+1}$ is a path in the visibility graph.
The points of~$A$ and~$C$ are positioned so that~$a_i$ and~$c_j$ see each other for all~$i,j$.

We now add a point~$b_i$ between and slightly below~$a_{\ell (i-1)}$ and~$a_{\ell (i-1) + 1}$, for each~$i \in [n]$.
The point~$b_i$ is not visible from any other~$a_j$, and adding it does not change the visibility in~$B \cup C$.
By adjusting the height of~$b_i$ relative to~$a_{\ell(i-1)}$, one can control the height of its `horizon'.
It is thus possible to position~$b_i$ to see $c_{\sigma(i)},\dots,c_n$, but not $c_1,\dots,c_{\sigma(i)-1}$.

The visibility graph of this terrain is thus as follows.
Between~$B$ and~$C$ is a half-graph, with edges~$b_ic_j$ whenever~$j \ge \sigma(i)$.
Furthermore $A = a_1,\dots,a_{\ell n}$ is a path, and~$b_i$ is adjacent to~$a_{\ell(i-1)},a_{\ell(i-1)+1}$ and no other vertex of~$A$.
Finally, $B$ is edgeless, $C$ is a clique, and between~$A$ and~$C$ is a complete bipartite graph.
We call~$G_\sigma^\ell$ this graph, and define the class of graphs for a function $f: \N \rightarrow \N$:
\[
    \Gc_f \coloneqq \Bigl\{G_\sigma^{f(n)} : n \in \N \text{ and } \sigma \text{ is a permutation on } [n]\Bigr\}.
\]

We show that for~$f$ sufficiently small, this class has unbounded twin-width.
The proof is similar to that of \cref{fact:logn-subdiv-tww}, except we transduce subdivided cliques rather than subdivided bicliques.
\begin{fact} \label{fact:logn-subdiv-tww2}
     For $f(n) = o(\log n)$ a positive function, $\Gc_f$ has unbounded twin-width.
\end{fact}

\begin{proof}
    We denote by $K_n^{\ell}$ the clique with $n$ vertices, where every edge is replaced by a path of length $\ell$. For a positive function $g$, consider the class $\Kc_g \coloneqq \{K_n^{g(n)} : n \in \N\}$. It is known that~$\Kc_g$ has unbounded twin-width whenever $g(n) = o(\log n)$ and $g(n) \ge 2$ \cite[Theorem~6.2]{twin-width2}. We show that $\Gc_f$ transduces $\Kc_g$ for $g(n) \coloneqq f(n^2) + 2 = o(\log n)$. Then $\Gc_f$ also has unbounded twin-width by \cite[Theorem~8.1]{Bonnet2022twinwidth1}.

    We construct an FO transduction $\Phi$ such that $\Phi(G_\sigma^\ell) \ni K_n^{\ell+2}$ for any $\ell \in \N$, where $\sigma$ will be chosen appropriately later on. Let $A, B, C$ be the partition of vertices of $G_\sigma^\ell$ as in its description. The vertices in $B$ are partitioned into $n$ parts $B_k = \{b_k^1, \dots, b_k^n\}$ for $k \in [n]$, each of which contains $n$ vertices. We choose every $n$-th vertex from $C$ to let $C' = \{c_1', \dots, c_n'\}$, which form a half-graph with the vertices in $B$ in such a way that $c_i'$ is connected with all the vertices in $B_j$ for $j \leq i$. 
    We want $b_j^i$ and $b_i^j$ to be connected by disjoint paths $P_{i,j}$ of length~$\ell$ whose inner vertices are in $A$, for $1 \leq i < j \leq n$.
    Since these $\binom{n}{2}$ paths are disjoint, we can choose a permutation $\sigma$ such that the path induced by~$A$ in $G_\sigma^\ell$ extends all of them.

    Now the transduction $\Phi$ is a composition of the following transductions, which is then again a transduction due to \cref{lem:transduction-comp}.
    Note that taking an arbitrary induced subgraph is a transductions, and so is removing all edges between two sets.
    \begin{itemize}
        \item Take the induced subgraph consisting of vertices in $C'$, $B_k \setminus \{b_k^k\}$ for $k \in [n]$, and $A' \coloneqq V(P_{i,j}) \cap A$ for $1 \leq i < j \leq n$.
        \item Remove all edges in the clique $G_\sigma^\ell[C']$, and in the biclique $G_\sigma^\ell[C' \cup A']$.
        \item Remove edges between $c_i'$ and the vertices in $B_j$ for $j < i$ by using an FO interpretation similar to the one used in \cref{fact:logn-subdiv-tww}.
    \end{itemize}
    Then the output of $\Phi$ is $K_n^{\ell+2}$.
\end{proof}

On the other hand, using \cref{lem:mw-long-paths}, and the fact that adding (in the sense of edge union) a single clique or biclique preserves bounded merge-width,
one can show that when~$f$ tends to infinity, the class~$\Gc_f$ has bounded merge-width.
It follows that for say $f(n) \eqdef \log \log n$, the class~$\Gc_f$ has bounded merge-width but unbounded twin-width, proving \cref{thm:terrain-not-delin}.

\subsection[H-graphs]{$H$-graphs}
Finally, we prove \cref{thm:Hgraphs-non-delin}: $H$-graphs are not delineated if~$H$ contains two cycles in the same connected component.

Let~$H_\infty$ be the graph consisting of a single vertex~$x_0$ with two self-loops.
For any bipartite graph $G = (U,V,E)$, let~$G^\bullet$ be the graph obtained by subdividing each edge of~$G$ once, and taking the complement of this subdivision, except that no edge is added between~$U$ and~$V$.
That is, the vertex set of~$G^\bullet$ is $U \uplus V \uplus E$, each $e \in E$ is adjacent to all of~$U,V$ except the two endpoints of~$e$,
and additionally each of~$U,V,E$ is a clique.
\begin{fact}
    For any bipartite graph~$G$, the graph~$G^\bullet$ is an $H_\infty$-graph.
\end{fact}
\begin{proof}
    Consider the following construction, depicted in \cref{fig:two-circles}.
    Take the topological realisation of~$H_\infty$: two circles~$C_1,C_2$ joined in exactly one point~$x_0$.
    Identify each vertex $u \in U$ with a half-circle~$X_u$ in $C_1 \setminus \{x_0\}$, picking arbitrary distinct half-circles for each vertex.
    Similarly, associate each $v \in V$ to some half-circle~$X_v$ in $C_2 \setminus \{x_0\}$.
    Then, for any edge $e = uv$ in~$E$, let~$X_e$ be the complement of $C_u \cup C_v$.
    Thus~$X_e$ consists of two half-circles which share the point~$x_0$, hence it is connected,
    and it clearly is disjoint from $X_u,X_v$, but intersects~$X_w$ for any $w \in (U \cup V) \setminus \{u,v\}$.
    The sets~$X_u$, $X_v$, $X_e$ give~$G^\bullet$ as intersection graph.
\end{proof}
\begin{figure}
    \centering
    \begin{tikzpicture}
        \def\r{2.5}
        \def\dt{0.08}
        
        \node[black,fill,draw,circle,inner sep=0pt, minimum width=1.5mm] (x0) at (0,0) {};
        \draw (\r,0) circle [radius=\r];
        \draw (-\r,0) circle [radius=\r];

        \pgfmathsetmacro{\rr}{\r+2*\dt}
        \draw (-\r,0) ++(90:\rr) -- (-1.2,3.2) node[fill=white] {$X_u$};
        \pgfmathsetmacro{\rr}{\r+4*\dt}
        \draw (\r,0) ++(110:\rr) -- (0.7,3.2) node[fill=white] {$X_v$};
        \pgfmathsetmacro{\rr}{\r-\dt}
        \draw (-\r,0) ++(50:\rr) -- (0,2.5) node[fill=white] (Xe) {$X_e$};
        \draw (\r,0) ++(130:\rr) -- (Xe);

        \foreach \i in {1,...,4}{
            \pgfmathsetmacro{\rr}{\r+\i*\dt}
            \pgfmathsetmacro{\a}{40+20*\i}
            \draw[teal, very thick] (-\r,0) ++(\a:\rr) arc [start angle=\a, delta angle=180, radius=\rr];
            \draw[magenta, very thick] (\r,0) ++(\a:\rr) arc [start angle=\a, delta angle=-180, radius=\rr];
        }

        \pgfmathsetmacro{\rr}{\r-\dt}
        \pgfmathsetmacro{\a}{80-3}
        \pgfmathsetmacro{\aa}{120+3}
        \draw[cyan, very thick] (-\r,0) ++(\a:\rr) arc [start angle=\a, delta angle=-174, radius=\rr];
        \draw[cyan, very thick] (\r,0) ++(\aa:\rr) arc [start angle=\aa, delta angle=174, radius=\rr];
        \draw[cyan, very thick] (-\dt,0) -- (\dt,0);
    \end{tikzpicture}
    \caption{%
        Construction of $G^\bullet$ as an $H_\infty$-graph.
        Half-circles~$X_u,X_v$ representing vertices of~$U,V$ are on either side in \textcolor{teal}{teal} and \textcolor{magenta}{magenta} respectively.
        In the middle in \textcolor{cyan}{cyan}: the subset $X_e$ representing some edge $e=uv$,
        chosen so that~$X_e$ intersects all half-circles on the left and right except~$X_u,X_v$.
    }
    \label{fig:two-circles}
\end{figure}
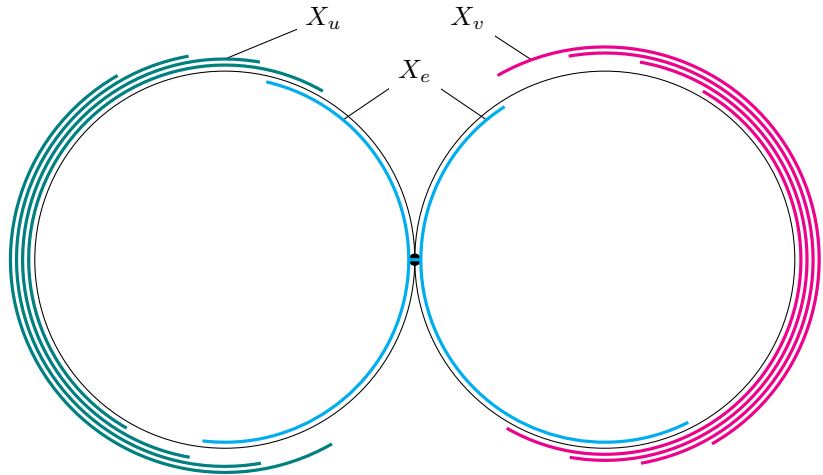

\begin{corollary}\label{cor:Hgraphs-non-delin}
    There is a hereditary class~$\Cc$ of $H_\infty$-graphs that is monadically dependent and has FPT FO model checking,
    but has unbounded merge-width (and thus unbounded twin-width).
\end{corollary}
\begin{proof}
    Take~$\Dc$ a class of bipartite graphs that is nowhere dense but has unbounded expansion (see e.g.~\cite{sparsity} for the definitions).
    These are typically obtained from graphs with girth and minimum degree both growing to infinity.
    By \cite[Corollary~1.8]{merge-width}, the class~$\Dc$ also has unbounded merge-width.
    Then the class of 1-subdivisions of graphs in~$\Dc$ is still nowhere dense (immediate from the definition of nowhere dense),
    and still has unbounded merge-width (because one can FO transduce~$\Dc$ from the class of 1-subdivisions of~$\Dc$, and FO transductions preserve bounded merge-width~\cite[Theorem~1.12]{merge-width}).
    Note that nowhere dense implies FPT FO model checking~\cite{grohe2017deciding} and monadically dependent~\cite{adler2014nowhereNIP}.
    
    Now let $\Cc = \{G^\bullet : G \in \Dc\}$ be obtained by the previous construction of $H_\infty$-graphs.
    The graph~$G^\bullet$ is obtained from the 1-subdivision of~$G$ by complementing edges between the appropriate sets of vertices.
    In the terminology of e.g.~\cite{flip-breakability}, this means that~$G^\bullet$ is obtained from the 1-subdivision by a bounded number of \emph{flips}.
    Flips are a very simple case of FO transduction, and thus preserve monadic dependence, FPT model checking, bounded twin-width, and bounded merge-width.
    It follows that~$\Cc$ still is monadically dependent and has FPT model checking, while having unbounded merge-width and twin-width.
\end{proof}

Finally, if~$H$ is any graph with two cycles (not necessarily disjoint) in the same connected component,
then~$H$ contains~$H_\infty$ as a minor (if we adapt the usual notion of minors to allow double edges and self loops).
It is simple to check that when~$H'$ is a minor of~$H$, then any $H'$-graph is also an $H$-graph.
Thus the non-delineated class~$\Cc$ of \cref{cor:Hgraphs-non-delin} can also be found as $H$-graphs, proving \cref{thm:Hgraphs-non-delin}.

\bibliography{literature}

\end{document}